\documentclass{LMCS}
\usepackage{hyperref,amsmath,amsfonts,amssymb,enumerate}  
\usepackage[latin1]{inputenc}
\usepackage{proof,stmaryrd,wrapfig,version}
\usepackage[all]{xy}

\usepackage{epsfig} 
\usepackage[latin1]{inputenc}

\theoremstyle{plain}
\theoremstyle{plain}\newtheorem{coin}[thm]{Coinduction Principle}

\renewcommand{\delta}{}

\newcommand{\mg}{\mbox{ $\!\rhd$}}
\newcommand{\mn}{\mbox{ $\!\lhd$}} 
\newcommand{\mge}{\mbox{ $\!\unrhd$}} 

\newcommand{\nmge}{\mbox{ \begin{picture}(4,4)
\put(0,-2){\makebox(4,2){$-$}}
\put(0,3){\makebox(4,2){$  \blacktriangleright$}}
\end{picture}
}} 

\newcommand{\bow}{\mbox{ $\!\bowtie\,$}}

\newcommand{\la}{\mbox{ $\lambda\hspace*{-1.3ex}\lambda$}}

 \newcommand{\reguno}[2]
  {
  $  \ \ \frac{\textstyle #1}{\textstyle #2} $
  }

\title[Conway Games, algebraically and coalgebraically]{Conway Games, algebraically and coalgebraically}

\thanks{Work supported by
 the FIRB Project RBIN04M8S8  (funded by MIUR), and
by the ESF Research Networking Programme GAMES}

\author[F.~Honsell]{Furio Honsell}	
\author[M.~Lenisa]{Marina Lenisa}	
\address{Dipartimento di Matematica e Informatica,
  Universit\`a di Udine\\
  via delle Scienze 206, Udine, Italy.}	
\email{furio.honsell@comune.udine.it, marina.lenisa@uniud.it}  


\def\doi{7 (3:08) 2011}
\lmcsheading%
{\doi}
{1--30}
{}
{}
{Mar.~30, 2011}
{Sep.~\phantom01, 2011}
{}

\begin{document}

\begin{abstract}
Using \emph{coalgebraic methods}, we extend Conway's  theory of games to  possibly  \emph{non-terminating},  \emph{i.e.} \emph{non-wellfounded games}  (\emph{hypergames}).
We take the view that a play which goes on forever is a \emph{draw}, and hence rather than focussing on 
winning strategies, we focus on \emph{non-losing strategies}.
Hypergames are a fruitful metaphor for non-terminating processes, \emph{Conway's sum}  being similar to \emph{shuffling}.
We develop a theory of hypergames, which extends in a non-trivial way Conway's theory; in particular, we generalize Conway's results  on game \emph{determinacy} and \emph{characterization} of strategies.
Hypergames have a  rather interesting theory, already in the case of
\emph{impartial hypergames}, for which we give a \emph{compositional semantics}, in terms of  a \emph{generalized Grundy-Sprague function} and a system of generalized \emph{Nim games}. 
\emph{Equivalences} and \emph{congruences} on games and hypergames are discussed.
We indicate a number of intriguing directions for future work.
We briefly compare hypergames with other notions of games used in computer science.
\end{abstract} 

\maketitle

\section{Introduction}
This paper arises from our attempt of understanding \emph{games} using  very  foundationally  unbiased tools, namely \emph{algebraic} and \emph{coalgebraic methods}. Of course, games arising in real life are extremely varied. They exhibit indeed a perfect example of a  \emph{family resemblance}
in the sense of Wittgenstein. Furthermore, in the past decades people have gone into the habit of describing, more or less conveniently,  an extremely wide gamut of interactions and other dynamic phenomena using game-based  \emph{metaphors}. 
To make matters even more complex when we speak about games many related  \emph{concepts} and  \emph{notions} come about, \emph{e.g.}  \emph{move}, \emph{position},  \emph{play},  \emph{turn}, \emph{winning condition},  \emph{payoff function}, \emph{tactics}, \emph{strategy}. None of these has a universal unique meaning, and according to the various presentations, these concepts are often blurred, sometimes taken as primitive sometimes explained and reduced to one another. And there are many more  \emph{properties}, which need to be specified more or less informally before  actually having pinned down the kind of game one is interested in; \emph{e.g.}  \emph{perfect knowledge}, \emph{zero-sum}, \emph{chance}, \emph{number} of players, \emph{finiteness}, \emph{determinacy}.

We think that Conway's approach to games provides a very elementary and sufficiently abstract  notion of game, which nonetheless is
significantly structured, because of the special r\^ole that \emph{sums} of games have in Conway's theory. 
And algebraic-coalgebraic methods provide a convenient conceptual setting for addressing this key concept. 

For these reasons, in this paper,
 we focus on Conway games \cite{Con76}, that is  \emph{combinatorial games}, namely no chance
2-player games,  the two players being conventionally called \emph{Left} (L)
and \emph{Right} (R). Such games have \emph{positions}, and in any position there
are rules which restrict L  to move to any of certain positions, called the 
\emph{Left positions},  while R may similarly move only to certain positions, called the  
\emph{Right  positions}. L and R move in turn, and 
the game is of \emph{perfect knowledge}, \emph{i.e.} all positions are public to both players. The game ends when one of the players has no move, the other player being the winner,  the \emph{payoff} function  yielding only 0 or 1.
 Many games played on boards
are combinatorial games, \emph{e.g.}  \emph{Nim},  \emph{Domineering},  \emph{Go}, 
\emph{Chess}. Games, like Nim,
where for every position both players have the same set of moves, are called \emph{impartial}. More general games, 
like
Domineering,  Go, 
Chess, where L and R 
may have different sets of moves are called \emph{partizan}.

Many other notions of games such as those which arise either in Set Theory, or in Automata Theory, or in Semantics of Programming Languages can be conveniently encoded in coalgebraic format, see \cite{HLR11a}. 

Here are some of the most frequently asked questions and corresponding arguments.

Which concept should be taken as primitive: moves or positions? We think that one of the strong points of Conway's format is that of focusing on positions. 
Actually, in Conway's theory, games, positions and moves essentially coincide. 
This approach is more general than others. 
Of course, the situation is much  nicer  when positions can be inductively  defined as finite sequences of more elementary tokens, namely moves, and hence positions can
be viewed as special kinds of plays, or vice versa.

Why do we need both a notion of  Player L  and  Player R, which in turn can be either  Player I, the first player, or  Player II, the respondent? Conway's choice arises from the fact that, as we pointed out above,  a key ingredient in his theory is that of \emph{sum} of games. 
In Conway's approach there is a rigid \emph{alternation} of moves between the two players in a sum game, but if we focus on a specific \emph{component} of a sum game this might no longer be so: the notions of first player and respondent in a subgame can change many times, and even a strict alternation of moves breaks down. 
This is in general why one can only refer meaningfully to player  L or R, notwithstanding the fact that any of the two can make an opening move in the game.

Of course many games where alternation of players is not rigid escape a direct encoding in Conway's games. However, in some cases these can be easily encoded in a coalgebraic format, as in some presentations of automata games, where  alternation of players is not assumed, see~\cite{HLR11a}. Sometimes,  encodings are
more roundabout, as in some card games, where the alternation between L and R is determined by a precise criterion, or as in \emph{normal form} games, arising in economic theory,  and in \emph{morra}, where the two players play \emph{simultaneously}.

Conway's approach clearly does not address the issue of a payoff function. 
Furthermore, since the winning condition is the absence of possible moves, \emph{i.e.} no next positions to reach for the player whose turn is on, some encoding is necessary to account for many games where winning or losing depends on the sequence of positions, or moves. There are many games of this kind. Some are rather silly, such as "My father is richer than yours`` where the two players in turn call a number and who call the largest is the winner. Some are less silly such as the one where the two players call two numbers in turn and the first player wins if the sum is equal to, say, 1 mod 4. Some are extremely important, such as those which arise in Set Theory, in connection with the Axiom of Determinacy, or in Automata Theory, where we have to deal, however, with infinite plays.

This is the last item we address: finite or infinite plays. One of the main contributions of this paper is the study of non-terminating (non-wellfounded) games, 
\emph{i.e.} games  on which plays are potentially infinite.  Especially in view of applications, potentially infinite interactions are even more important than finite ones.
The importance of games for Computer Science comes from the fact that they capture in a natural way the notion of \emph{interaction}. Non-wellfounded games model in a faithful way \emph{reactive processes} (operating systems, controllers, communication protocols, 
etc.), that are characterised by their 
\emph{non-terminating} behaviour and perpetual \emph{interaction} with their environment.

We take the ``natural'' view that all infinite plays are \emph{draws}: on infinite plays, apparently, there are no losers, because each player can
respond indefinitely. This naturally extends the winning condition on Conway's games.
In this paper, 
in fact, we shall address only games were all infinite plays are draws,  \emph{i.e.} 
 ``free'' games in Conway's terminology; 
we will not deal with  games where certain infinite plays are set to be draws and others to be winning for one of the two players, namely ``mixed'' games,  or with
``fixed'' games, where
infinite plays are all winning for one of the two players, see \emph{e.g.} \cite{BCG82} for more details. 
Since we take non terminating plays to be draws,  the notion of \emph{winning strategy} on hypergames has to be replaced by that of 
\emph{non-losing strategy}.

Combinatorial Game Theory started at the beginning of 1900 with the study of the famous impartial game Nim, which became also a movie star in the 60s, in the film 
``L'ann\'ee derni\`ere \`a Marienbad" by Alain Resnais and 
Alain Robbe-Grillet.
 In the 1930s,
Sprague and Grundy  generalized the results on Nim to all impartial  terminating  (\emph{i.e.} well-founded)  games, \cite{Gru39,Spra35}. 
In the 1960s, Berlekamp, Conway, Guy introduced the theory of partizan games, which  first appeared  in the
 book  ``On Numbers and Games''  \cite{Con76}. In \cite{Con76}, the theory of games is connected to the theory of
\emph{surreal numbers}.

In \cite{Con76}, the author focussed essentially  on   terminating games, \emph{i.e.} games on which all plays are finite.
Non-terminating games
were intentionally neglected  as ill-formed or trivial  games,  not interesting for ``busy  men'', and their discussion was confined to a single chapter, inspired
by~\cite{Smi66}.
Non-wellfounded games have been later considered in \cite{BCG82}, Chapters 11-12, were free, fixed and mixed games have been discussed. 
In these chapters,
the authors also consider an interesting generalization 
of the Grundy-Sprague theory, originally due to Smith  \cite{Smi66},
which provides natural notions of \emph{canonical forms}. However, not much attention has yet been paid to generalize the results in \cite{Con76} to non-terminating free games.

Possibly non-terminating games, which we call \emph{hypergames}, can be naturally defined as a \emph{final coalgebra} of \emph{non-wellfounded sets} (\emph{hypersets}), which are 
the sets of  a universe of Zermelo-Fraenkel satisfying a suitable
Antifoundation Axiom, see \cite{FH83,Acz88}. 
This definition generalizes directly the original one in~\cite{Con76}, where games are taken to be \emph{well-founded sets}. 
Once hypergames are defined as a final coalgebra,  
 operations on Conway's games, such as \emph{disjunctive sum},  can be  naturally extended to hypergames, by
defining them as \emph{final morphisms} into the coalgebra of hypergames. 

Our approach is different 
from other approaches in the literature, where games are defined as \emph{graphs} or \emph{pointed graphs}. Viewing games as sets (or points of 
a final coalgebra) allows us to abstract away from superficial features of positions and to reason directly up-to graph bisimilarity. 
Our approach is justified 
by the fact that all important properties of games in Conway's setting, \emph{e.g.} existence of winning/non-losing  strategies, are invariant under bisimilarity.

Our theory of hypergames generalizes the original  theory on Conway's games of 
\cite{Con76} rather smoothly, but significantly.
Two important results in our paper are  \emph{Determinacy} and a  \emph{Characterization Theorem} of non-losing strategies on hypergames. 
The latter requires (a non-trivial) generalization  of Conway's partial order relation on games to
hypergames.   

On top of bisimilarity, various notions of \emph{equivalences} and \emph{congruences} on games and hypergames, arising by looking at strategies, are studied in this paper. In particular, 
we investigate various characterizations of the
\emph{greatest congruence} w.r.t. sum, refining the \emph{equideterminacy} relation.
One interesting result of our investigation is that this congruence coincides on Conway's games with the equivalence induced by Conway's partial order, and with
the (extended) \emph{Grundy semantics} on impartial (hyper)games. 

For the  class of impartial hypergames,
we revisit and extend  in a coalgebraic setting the theory of Grundy-Sprague and Smith  based on the canonical \emph{Nim games}, by introducing
suitable \emph{canonical $\infty$-hypergames}. We show that such canonical hypergames 
can be construed in our setting as a truly 
 \emph{compositional semantics} of impartial hypergames, \emph{fully abstract} w.r.t. the greatest behavioral congruence. Such semantics is given via a suitable \emph{generalized Grundy function}, which we define on the
whole class of hypergames. Our approach extends other approaches in the literature, where the generalized Grundy function is defined only on certain classes of finite 
cyclic graphs, see \emph{e.g.} \cite{FR01}. 
\medskip

This paper is a revised and extended version of \cite{HL09}. The parts on game and hypergame equivalences, and semantics are new. The section on impartial hypergames has been  substantially revised and extended.
 
\subsubsection*{Summary.}  In Section~\ref{tcg}, we provide a  presentation of Conway's games as an initial algebra, and 
we introduce  hypergames as a final coalgebra for the same functor. We discuss determinacy of games, and we prove 
characterization
theorems for winning and non-losing strategies. In Section~\ref{gcon},  we give coalgebraic definitions of sum and negation on hypergames, which extend
original Conway's definitions. In  Section~\ref{ges}, we study equivalences on games, in particular we introduce and study the notion of contextual equivalence. 
In Section~\ref{img}, we investigate  impartial hypergames. In particular, we extend the Grundy-Sprague theory, and we define a generalized Grundy function, which 
 gives a compositional semantics for hypergames, fully abstract w.r.t. contextual equivalence.  Comparison with related
games and directions for future work appear in Section~\ref{final}.

\subsubsection*{Acknowledgements.} We would like to thank the anonymous referees for many useful comments, which helped in improving the paper.  

\section{From Conway's Games to Hypergames}\label{tcg}
In this section, first we present Conway's games as an \emph{initial algebra} for a suitable functor, then we introduce
hypergames as a \emph{final coalgebra} for the same functor. We take infinite plays to be \emph{draws}, and hence
Conway's notion of \emph{winning strategy} has to be generalized by that of \emph{non-losing strategy}. In this 
section, we present fundamental results on hypergames, generalizing corresponding results on Conway's games.
In particular, we discuss \emph{determinacy} of games, \emph{i.e.} existence of winning/non-losing strategies, and we prove \emph{characterization
theorems} for winning and non-losing strategies, which extend in a non-trivial way corresponding results on Conway's games.
\medskip

We recall that Conway games are 2-player games, the two players are called \emph{Left} (L)
and \emph{Right} (R). Such games have \emph{positions}, and in any position $x$ there
are rules which restrict Left to move to any of certain positions, called the 
\emph{Left positions} of $x$, while Right may similarly move only to certain positions,  called the 
\emph{Right positions} of $x$.  Since we are interested only in the abstract structure of games,
we can regard any position $x$ as being completely determined by its Left and Right options, and
we shall use the notation $x= (  X^L, X^R )$,  
where $X^L, X^R$ denote sets of positions.
Games are identified with their initial positions, and they can be represented as the tree of all positions generating from the initial one. 
Left and Right  move in turn, and the game 
ends when one of the two players does not have any option.  All games are terminating, \emph{i.e.} infinite sequences of moves cannot arise. However, there can be  possibly infinite moves at any position.
\medskip

\begin{rem}
Contrary to other notions of games, where only player I (the player who starts the game) and  player II  are considered, Conway distinguishes also
between L an R.
Both the case where L starts the game, \emph{i.e.} he acts as player I, while R acts as player II, and the case where L acts as player II and R as player I are considered. This extra complexity allows for a  definition of the \emph{sum operation on games}, which is central to the theory of Conway games, and  which is such that the alternance of player I and player II  can break in any single component. Thus the need of considering, at each step, all the possible moves of both  L and R on the games where L and R have different sets of moves.
\end{rem}

Conway's games can be viewed as an \emph{initial algebra} of
a suitable functor.   Although such games are well-founded,  in view of extensions to non-wellfounded games, 
we work in the category  $\mbox{Class}^*$ of classes of  possibly non-wellfounded sets (hypersets) and functional 
classes.\footnote{Non-wellfounded sets are
the sets of  a universe of Zermelo-Fraenkel satisfying the
Antifoundation Axiom, see \cite{FH83,Acz88}. Alternatively to classes of sets, we could consider an \emph{inaccessible cardinal} $\kappa$, and the category whose objects are the sets with \emph{hereditary cardinal} less than ${\kappa}$, and whose morphisms are the functions with hereditarily cardinal less than ${\kappa}$.  We recall that the \emph{hereditary cardinal}  of a set is the cardinality of its transitive closure, namely  the cardinality of the downward membership tree which has the given set at the root.}

\begin{defi}[Conway Games]  The set of \emph{Conway's Games} ${\mathcal {G}}$ is inductively defined by
\begin{iteMize}{$\bullet$}
\item the empty game $(\emptyset , \emptyset)\in {\mathcal G}$;
\item if $X,  X' \subseteq {\mathcal G}$, then $( X,  X')\in {\mathcal {G}}$.
\end{iteMize}
Equivalently, ${\mathcal {G}}$
 is the carrier of the  \emph{initial algebra} $({\mathcal {G}}, \mathit{id})$  
 of the functor $F:  \mbox{Class}^* \rightarrow \mbox{Class}^*$,
defined by $F(A) = {\mathcal P}(A) \times {\mathcal  P}(A)$, where ${\mathcal  P}(A)$ is the \emph{powerset} functor (with usual definition on morphisms).
\end{defi}

\noindent \emph{Notation}. 
Games will be denoted by small letters, \emph{e.g.} $x$, with $x=(X^L,X^R)$, and $x^L, x^R$ will denote generic elements of $X^L,X^R$. 
We denote by  $\mbox{Pos}_x$ the set of \emph{positions} hereditarily reachable from $x$.
\medskip

\noindent \emph{Some simple games.} The simplest  game is the empty one, \emph{i.e.} $(\emptyset , \emptyset)$, which will be denoted  by $0$.  Then
 we define the games $1= (\{ 0\} , \emptyset)$, $-1 =( \emptyset  , \{0\})$, $*1= (\{ 0\} , \{ 0\})$. Intuitively, in the game $0$, the  player who starts will lose (independently whether he plays L or R), since there are
 no moves. Thus player II  has a winning strategy. In the game $1$ there is a winning strategy for L, since, if L plays first,
 then  L has a move to $0$, and R has no further move; otherwise, if R plays first, then he loses, since he has no moves.
 Symmetrically, $-1$ has a winning strategy for $R$. Finally, the game $*1$ has a winning strategy for player I, since
 he has a move to $0$, which is losing for the next player. Intuitively, a strategy for a given player is a function which,
 for any position where the player is next to move, gives, if any, a move for  this player.
 A strategy is winning if it provides answers against  any strategy for the opponent player. The notions of strategy and winning strategy are 
  formalized in Section~\ref{strat} below.
 \medskip 

Hypergames  can be naturally defined as a \emph{final coalgebra} on possibly non-wellfounded sets:

\begin{defi}[Hypergames] 
 The set of \emph{hypergames} ${\mathcal {H}}$ is the carrier of the \emph{final  coalgebra} $({\mathcal {H}} , \mathit{id})$ of the functor 
 $F:  \mbox{Class}^* \rightarrow \mbox{Class}^*$.
\end{defi}
Hence hypergames subsume  Conway's games. In the sequel, we will often refer to hypergames simply as games.

Defining hypergames as a final coalgebra, we immediately get a \emph{Coinduction Principle} for reasoning on possibly non-wellfounded games:

\begin{lem}\label{hbis}
A \emph{$F$-bisimulation on the coalgebra $({\mathcal {H}} , \mathit{id})$} is a symmetric relation
${\mathcal R}$ on hypergames such that, for any $x=(X^L, X^R),\ y=(Y^L, Y^R)$, 
\[ x {\mathcal R} y \ \Longrightarrow \ (\forall x^L \in X^L.\exists y^L \in Y^L. x^L  {\mathcal R}y^L) \ \wedge\ 
(\forall x^R \in X^R.\exists y^R \in Y^R. x^R {\mathcal R}y^R) \ .\]
\end{lem}

\begin{coin}\label{cbas}
Let us call a  $F$-bisimulation on  $({\mathcal {H}} , \mathit{id})$ a \emph{hyperbisimulation}. The following principle holds:
\[   \frac{{\mathcal R} \mbox{ hyperbisimulation} \ \ \ x{\mathcal R}  y}{x=y} \]
\end{coin} 

Most notions and constructions on games taken as graphs are  invariant w.r.t. hyperbisimilarity. In particular, hyperbisimilar games 
will  be \emph{equidetermined}, \emph{i.e.} they will have winning/non-losing strategies for the same players.
Hypergames correspond to graphs taken up-to bisimilarity;  the coalgebraic representation  naturally induces a minimal representative for each 
bisimilarity equivalence class. For instance, all game graphs with no-leaves are represented by the hypergame $x=\{ x\}$.
\medskip

\noindent \emph{Some simple hypergames.} Let us consider the following pair of simple hypergames:
$a= (\{ b\}, \emptyset )$ and $b=(\emptyset , \{ a\} )$. If L plays as II on $a$, then he immediately wins since R has no move.
If L plays as I, then he moves to $b$, then R moves to $a$ and so on, an infinite play is generated. This is a draw.
Hence L has a non-losing strategy on $a$. Simmetrically, $b$ has a non-losing strategy for R. 
Now let us consider the hypergame $c=(\{c\}, \{ c\})$. On this game,  any player (L,R,I,II) has a non-losing 
strategy; namely there is only the non-terminating play consisting of infinite $c$'s.

\subsection{Strategies}\label{strat}
Before giving the formal definition of strategy, we introduce the notion of \emph{play} over a game as an alternating sequence of positions on the game, starting from 
the initial position, and we define winning and non-losing plays:
 
 \begin{defi}[Plays]\hfill 
 \\
 \phantom{ii}(i) A \emph{play} on a game $x$ is a (possibly empty) finite or infinite sequence of positions $ \pi=  x_1^{K_1} x_2^{K_2} \ldots$  such that 
 \begin{iteMize}{$\bullet$}
\item $\forall i.\ K_i \in \{ L, R\}$;
\item  $x=x_0$  and $\forall i\geq 0\ (x^{K_i}_i =(X_i^L, X_i^R)\ \wedge\ x^{K_{i+1}}_{i+1}\in X_i^{\overline{K}_i})$,\ \ \ 
 where $\overline{K}_i=\begin{cases} L & \mbox{ if } K_i=R
 \\ R & \mbox{ if } K_i=L
 \end{cases}$
 \end{iteMize}
 \noindent   We denote  by $\mathit{Play}_x$ the set plays on $x$,  by $\mathit{FPlay}_x$ the set of finite plays on $x$, and by $\epsilon$ the empty play.
 \\ \phantom i(ii) A play $\pi$ is \emph{winning} for player L (R)  iff it is finite and it ends with a position $y = (Y^L, Y^R)$ where R (L) is next to move but $Y^R=\emptyset$ ($Y^L=\emptyset$).
 We denote by $\mathit{WPlay}^L_x$ ($\mathit{WPlay}^R_x$) the set of plays on $x$ winning for L (R).
 \\ (iii) A play  $\pi$ is a \emph{draw} iff it is infinite. We denote by $\mathit{DPlay}_x$ the set of draw plays.
 \\ (iv) A play $\pi$ is \emph{non-losing} for player L (R) iff it is winning for L (R) or it is  a draw, \emph{i.e.} we define $ \mathit{NPlay}^L_x = \mathit{WPlay}^L_x
 \cup \mathit{DPlay}_x$ ($ \mathit{NPlay}^R_x = \mathit{WPlay}^R_x
 \cup \mathit{DPlay}_x$).
 \end{defi}
 
Strategies for a given player can be formalized as functions on plays ending with a move of the opponent player, telling, if any,  which is the
next move of the given player.
In what follows, we denote by 
\begin{iteMize}{$\bullet$}
\item $\mathit{FPlay}^{LI}_x$  the set of finite plays on which L acts as  player I, and  ending  with a position where L  is next to move, \emph{i.e.} $\mathit{FPlay}^{LI}_x= \{ \epsilon \} \cup \{ x_1^{K_1} \ldots x_n^{K_n} \in \mbox{FPlay}_x\  |\   K_1=L\ \wedge\ K_n =R,\  n> 1 \}$
\item $\mathit{FPlay}^{LII}_x$  the set of finite plays on which L acts as player II, and  ending  with a position where L  is next to move, \emph{i.e.}
$\mathit{FPlay}^{LII}_x = \{ x_1^{K_1} \ldots x_n^{K_n}  \in \mbox{FPlay}_x\  |\   K_1=R\ \wedge\ K_n =R,\ n\geq 1 \}$. 
\item Similarly we define $\mathit{FPlay}^{RI}_x$, $\mathit{FPlay}^{RII}_x$.
\end{iteMize}

We define:

\begin{defi}[Strategies] \label{stp}
Let $x$ be a game. \\
(i) A strategy $f$ for LI (\emph{i.e.} L acting as player I) is a partial function $f : \mathit{FPlay}^{LI}_x \rightarrow \mathit{Pos}_x$ such that, for any
$\pi\in \mathit{FPlay}^{LI}_x$,
\begin{iteMize}{$\bullet$}
\item $f (\pi)=  x'  \ \Longrightarrow \ \pi  x' \in \mathit{FPlay}_x$
\item $\exists x' .\ \pi x'\in \mathit{FPlay}_x  \ \Longrightarrow \ \pi\in \mathit{dom}(f)$.
\end{iteMize}
\noindent Similarly,  one can define strategies for players LII, RI, RII.
\\ (ii) Moreover,  we define:
\begin{iteMize}{$\bullet$}
\item a strategy for player L is a pair of strategies for LI and LII, $f_{LI} \uplus f_{LII}$;
\item a strategy for player  R is a pair of strategies for RI and RII, $f_{RI} \uplus f_{RII}$;
\item a strategy for player  I is a pair of strategies for LI and RI, $f_{LI} \uplus f_{RI}$;
\item a strategy for player  II is a pair of strategies for LII and RII, $f_{LII} \uplus f_{RII}$.
\end{iteMize} 
\end{defi}

Strategies, as defined above, provide answers (if any) of the given player on \emph{all} plays ending with a position where the player is next to move.
Actually, we are interested only in the behavior of a strategy on those plays which arise when it interacts with (counter)strategies for the opponent player.
Formally, we define:

\begin{defi}[Product of Strategies] Let $x$ be a game, and  P a player  in $\{$LI,LII,RI,RII$\}$.
\hfill
\\ \phantom i(i) Let $\pi$ be a play on $x$, and $f$ a strategy on $x$ for P. We say that $\pi$ is \emph{coherent with} $f$ if, for any proper prefix $\pi'$ of $\pi$ ending with a position where player P is next to move,
\[ f (\pi') = x' \ \Longrightarrow \ \pi' x' \mbox{ is a prefix of } \pi\ .\]
\noindent (ii)  Given a strategy $f$ for P  on $x$, and a counterstrategy $f'$, \emph{i.e.} a strategy for the opponent player, 
we define the \emph{product} of $f $ and $f'$, $f *f'$, as the unique play coherent with both $f$  and $f'$.
\end{defi}

Now we are ready to define \emph{non-losing}/\emph{winning strategies}.
Intuitively, a strategy is 
 non-losing/winning for a player, if it generates non-losing/winning plays against any  possible counterstrategy. 
 
 \begin{defi}[Non-losing/winning Strategies]\label{wns}
Let $x$ be a game, and P a player in $\{$LI,LII, RI,RII$\}$.
  \\ (i) A strategy $f$ on $x$ is \emph{non-losing} for P if,   for any strategy $f'$ on $x$ for the opponent player, 
   $f* f'\in \mathit{NPlay}^{P}_x$.
   \\ (ii) A strategy $f$ on $x$ is \emph{winning} for P if,   for any strategy $f'$ on $x$ for the opponent player, 
   $f* f'\in \mathit{WPlay}^{P}_x$.
   \\ (iii) A strategy $f_{LI} \uplus f_{LII}$ for player L is \emph{non-losing/winning} if $f_{LI}$ and $ f_{LII}$ are non-losing/winning strategies for LI and LII, respectively.
   Similarly for players R,I,II.
    \end{defi}
 
  Notice that on Conway's games, where infinite plays do not arise, the notion of non-losing strategy coincides with that of winning strategy.\medskip
   
  Intuitively, the winning condition on  finite  plays,``no more moves for the next player in the current position", does not depend on the ``history'', \emph{i.e.} on the whole sequence of positions, but only  on the the last position. Hence, having taken all infinite plays to be draws, one can prove that,
 for any non-losing/winning strategy on a game $x $, there exists a \emph{positional} (\emph{history-free}) non-losing/winning strategy on $x$. Formally, we define:
   
   \begin{defi}[Positional Strategies] Let P be a player in $\{$LI,LII,RI,RII$\}$.
 \\ (i)  A \emph{positional strategy}  $f$ for P on $x$ is a  strategy for P such that, for all $\pi x', \pi'x'\in  \mathit{FPlay}^{P}_x$,
  either $f$ is not defined on both $\pi x'$ and $ \pi' x'$,  or   $f(\pi x')=f(\pi' x')$.
 \\ (ii) A \emph{positional non-losing/winning strategy} for P is a positional strategy which is non-losing/winning for P.
    \end{defi}
 
 \begin{prop}\label{posi}  Let P be a player in $\{$LI,LII,RI,RII$\}$.
If there exists a non-losing/winning strategy $f$ for P on $x$, then there exists a positional non-losing/winning strategy for P on $x$.
 \end{prop}
  \begin{proof}
  Let $f:\mathit{FPlay}^{P}_x   \rightarrow  \mbox{Pos}_x$ be a strategy for P on $x$. Then we can define a positional strategy $\overline{f}$
  for P  as follows.
  \\ (i) $ \overline{f} (\epsilon)= f(\epsilon)$.
\\ (ii) For any $x_i$  such that there exists $\pi.\  \pi x_i \in \mathit{dom}(f)$, by the Axiom of Choice, we can  choose an element $\overline{x}_i$ in 
$\Pi_{x_i} = \{ f(\pi x_i) \mid 
\pi\in \mathit{FPlay}^{P}_x \ \wedge\ 
\pi x_i \in \mathit{dom}(f)\ \wedge\ \pi x_i \mbox{ coherent with } f \}$, if $\Pi_{x_i} \neq \emptyset$, otherwise we choose $\overline{x}_i$ in 
$\{ f(\pi x_i) \mid 
\pi\in \mathit{FPlay}^{P}_x \ \wedge\ 
\pi x_i \in \mathit{dom}(f) \}$.
 \\  Clearly 
  $\overline{f}$
  is positional. Moreover, $\overline{f}$ is non-losing/winning iff $f$ is.
  \end{proof} 
   
   
   As a consequence of the above proposition, we can restrict ourselves to considering only positional strategies. 

   \subsection{Determinacy Results}
   On Conway's games a strong determinacy result holds, \emph{i.e.} any game has a winning strategy for exactly one player in $\{$L,R,I,II$\}$. This
   does not hold on hypergames, where 
we can have  non-losing strategies for various players at the same time, as in the case of the game $c$ above. However, on any hypergame there exists a 
non-losing strategy for \emph{at least} one player. Moreover, if there is a winning strategy for a given player, then there are no non-losing strategies for the
other players. This subsumes Conway's determinacy result. In what follows, we formalize the above results.

The following lemma is instrumental: 

\begin{lem} \label{non}
Let $x$ be a game. 
 \\ \emph{\phantom{ii}(i)}
 LI has a winning strategy on $x$ iff RII does not have a non-losing strategy on  $x$.
\\ \emph{\phantom i(ii)} LII has a winning strategy on $x$ iff RI does not have a non-losing strategy on  $x$.
\\ \emph{(iii)} Symmetrically, exchanging the r\^ole of L and R.
\end{lem}
\begin{proof} \hfill
\\  (i) $(\Rightarrow)$ Let $f$ be a winning strategy for LI. Then, by definition,  for any strategy $f'$ for RII, $f*f'$ is winning for L. Hence RII cannot have any non-losing
strategy. 
\\ (i) $(\Leftarrow )$ Assume RII has no non-losing strategy. Then we can build a strategy $f: \mathit{FPlay}^{LI}_x \rightarrow \mathit{Pos}_x$ for LI by induction on finite plays, with the property that, for any play $\pi\in \mathit{FPlay}^{LI}_x$ coherent with $f$, $f(\pi)$ is defined  and RI has no non-losing strategy from $\pi f(\pi)$. Namely,
since RII has no non-losing strategy on $x$, then there exists an opening L move such that RI has no non-losing strategy from that position. This allows us to define 
$f$ on the empty play. Now assume to have defined $f$ on plays of length $n$. Let us consider a play $\pi$ of length n+1 coherent with $f$. By induction hypothesis,
RI has no non-losing strategy from $\pi$, hence for any  R  move  from $\pi$ to a position $x'$, L has an answer bringing to a position where RI has no
non-losing strategy. This allows us to extend $f$ on all plays of length $n+2$ coherent with $f$; on plays of length $n+2$ not coherent with $f$ but extensible with
a L move, we can define $f$ in an arbitrary way. This gives a strategy for LI, which is winning, because, by definition of $f$, for any counterstrategy $f'$ for RII,
$f*f'$  cannot be non-losing for R.
\\ (ii)-(iii) The proofs are similar  to the above one.
\end{proof}

\begin{thm}[Determinacy]\label{nls} 
Any  game has a  non-losing strategy at least for one of the players L,R,I,II.
\end{thm}
\begin{proof}
Assume by contradiction that $x$ has no non-losing strategies for L,R,I,II. Then in particular $x$ has no non-losing strategy for LI or
for LII. Assume the first case holds (the latter can be dealt with similarly). Then, by Lemma~\ref{non}, $x$ has a winning
strategy for RII. Now, since by hypothesis R has no non-losing strategy,  then there is no non-losing strategy for RI. But then, by Lemma~\ref{non}, there is a winning
strategy for LII. Therefore, by definition, there is a winning strategy for II. Contradiction.
\end{proof}

Theorem~\ref{nls} above can be sharpened, by considering when the non-losing strategy  is in particular a winning strategy:

\begin{thm}\label{qua}
Let $x$ be a game. Then either there exists a winning strategy on $x$ for exactly one of the players  L,R,I,II, and there are no non-losing strategies for the other players;
or at least two of L,R,I,II have a non-losing strategy. 
In this latter case either \emph{(i)} or \emph{(ii)} holds:
\\ \emph{\phantom i(i)} \emph{either} L or R have a non-losing strategy and  \emph{either} I or II have a non-losing strategy;
\\ \emph{(ii)} \emph{all}  players  L,R,I,II have a non-losing strategy.
\end{thm} 
\begin{proof} Assume \emph{e.g.} L has a winning strategy on $x$. Hence, by Lemma~\ref{non}, both RI and RII have no non-losing strategies. Therefore, R,I,II do not have non-losing strategies.
\\ Otherwise, assume \emph{e.g.} L has a non-losing strategy but no winning strategies on $x$. Then three cases can arise: (1) both
LI and LII have no winning strategies; (2) LI has no winning strategy, but LII has a winning strategy; (3) LI has a winning strategy, but LII
has no winning strategy. In the first case, by Lemma~\ref{non} both RII and RI have a non-losing strategy.
Hence all players, L,R,I,II, have non-losing strategies.
 In the second case, by 
Lemma~\ref{non} RII has a non-losing strategy, but  RI has no non-losing strategy, hence player II has a  non-losing strategy, while both R and I have no non-losing strategies.  Using a similar argument, one can show that in the third case, I has a non-losing strategy, but both R and II have no non-losing strategies.
\end{proof}

According to Theorem~\ref{qua} above, the space of hypergames can be decomposed as in Figure~\ref{fig}.
For example, the game $c=(\{ c\}, \{ c\})$ belongs to the center of the space, because it has non-losing strategies for all players,
while the games $a= (\{ b\},\{ \})$
and $b=(\{ \}, \{ a\})$ belong to the sectors marked with L,II and R,II, respectively, and the games $a_0 =(\{ b_0 \}, \{ 0\})$
and $b_0=(\{0 \}, \{ a_0\})$ belong to the sectors marked with R,I and L,I.

\begin{figure}
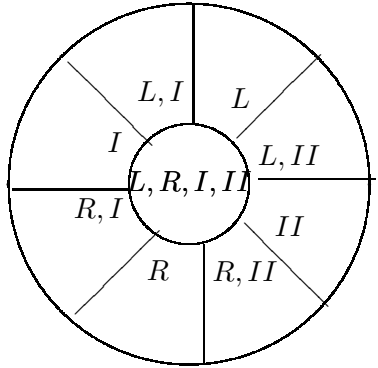

\[
\xy 
(0,30)*{ L,R,I, II}*\xycircle(8,8){-}="q0";
(0,30)*{ L,R,I, II }*\xycircle(24,24){-}="q1";
(-24,30)*{.}="p1";
(24,30)*{.}="p2";
(0,54)*{.}="p3";
(0,6)*{.}="p4";
(-17,47)*{.}="p5";
(17,47)*{.}="p6";
(-17,13)*{.}="p7";
(17,13)*{.}="p8";

{\ar @{-} "q0"; "p1" <2pt> ^{\txt{$R,I$}}}
\ar @{-} "q0"; "p2" <2pt> ^{\txt{$L,II$}}

\ar @{-} "q0"; "p3" <2pt> ^{\txt{$L,I$}}
\ar @{-} "q0"; "p4" <2pt> ^{\txt{$R,II$}}

\ar @{-} "q0"; "p5" <2pt> ^{\txt{$I$}} 

\ar @{-} "q0"; "p6" <2pt> ^{\txt{$L$}}

\ar @{-} "q0"; "p7" <2pt> ^{\txt{$R$}}
\ar @{-} "q0"; "p8" <2pt> ^{\txt{$II$}}

\endxy
\]
\caption{The space of hypergames.}
\label{fig}
\end{figure}
 
As a corollary of Theorem~\ref{qua} above we get Conway's determinacy result:

\begin{thm}[Determinacy, \cite{Con76}]\label{wst}
Any Conway's game has a winning strategy either for L or for R or for I or for II.
\end{thm}

As a consequence, Conway's games are all in the sectors L,R,I,II of Figure~\ref{fig}. However, notice that these sectors do not contain
only well-founded games, \emph{e.g.} the non-wellfounded game $d= (\{ d'\}, \{ d'\})$, where $d'= (\{d', 0\}, \{ d',0\})$ has a winning strategy for II.

\subsection{Characterization Results}
In \cite{Con76}, a relation $\gtrsim$ on games  is introduced, inducing a partial order (which is a total order on the subclass of games
corresponding to surreal numbers). Such relation allows to characterize Conway's games with a winning strategy for L,R,I or II.
In what follows, first we recall the above results on Conway's games, then we show how to generalize them to hypergames.
This generalization is based on a non-trivial extension of the relation $\gtrsim$.

The definition of $\gtrsim$ might appear a little strange at first. However, its structure is ultimately easy to grasp if we restrict to so called 
\emph{surreal numbers}, \emph{i.e.} games where all L members are hereditarily non-$\gtrsim$ of any R member. The definition then is akin to the definition 
of Dedekind section on real numbers.  The only difference lying in the fact that surreal numbers  are defined inductively, rather than ``impredicatively'', as in the
case of real numbers.

\begin{defi}[\cite{Con76}] \label{geq}
Let $x=(X^L, X^R)$, $y=(Y^L, Y^R)$ be Conway's games.
We define, by induction on games: \[ x\gtrsim y\ \  \mbox{iff}\ \  \forall x^R \in X^R.\  (y\not\gtrsim x^R) \ \wedge\  \forall y^L \in Y^L.\   (y^L \not\gtrsim x) \ .\]
Furthermore, we define:
\\ -- $x>y\ \  \mbox{iff}\ \  x\gtrsim y \ \wedge\ y\not \gtrsim x$
\\ -- $x\sim y\ \  \mbox{iff}\ \  x\gtrsim y \ \wedge\ y\gtrsim x$
\\ -- $x || y \ (x\mbox{ fuzzy } y) \ \  \mbox{iff}\ \  x\not\gtrsim y \ \wedge\ y\not\gtrsim x$
\end{defi}
Notice that $\not\gtrsim $ does not coincide with $<$, \emph{e.g.} $*1= (\{ 0\}, \{ 0\})$ is such that $*1 \not \gtrsim 0$ holds, but 
$*1 \gtrsim 0$ does not hold.
As one may expect, $1> 0>-1$, while for the game $*1$ (which is not a number), we have $*1 || 0$. 
Moreover:

\begin{prop}[\cite{Con76}]\label{ceq}
$\sim$ is an equivalence relation.
\end{prop}

However, notice that $||$ is \emph{not} an equivalence relation, since $||$ is \emph{not}  reflexive:  trivially it is not the case that $0 || 0$.

 The following important theorem gives the connection between Conway's games and numbers, and it allows to characterize games according to winning strategies:

\begin{thm}[Characterization,  \cite{Con76}] \label{car}
Let $x$ be a Conway's game. Then
\\
\begin{tabular}{llll}
$x>0 $ & ($x$ is positive)\ \ \ \   & iff\ \ \ \ & $x$ has a winning strategy for L.
\\  $x<0 $\ \  & ($x$ is negative) & iff & $x$ has a winning strategy for R.
\\ $x\sim 0 $ & ($x$ is zero) & iff & $x$ has a winning strategy for II.
\\ $x||0 $ & ($x$ is fuzzy) &  iff & $x$ has a winning strategy for I.
\end{tabular}
\end{thm}

The  generalization  to hypergames of Theorem~\ref{car} above is quite subtle, because it requires 
 to extend the relation $\gtrsim$ to hypergames, and this needs particular care. We would like to define
such relation  \emph{by coinduction}, as the greatest  fixpoint of a monotone operator on relations, however the 
operator which is naturally induced by the definition of $\gtrsim$ on Conway's games is \emph{not} monotone. 
This problem can be overcome as follows. 

Observe 
that the relation $\gtrsim$ in Definition~\ref{geq}  is defined in terms of the relation $\not\gtrsim$. Vice versa $\not \gtrsim$ is defined in terms of $\gtrsim$ by:
\[ x\not\gtrsim y \ \mbox{iff}\ \exists x^R.  y \gtrsim x^R \ \vee \ \exists y^L. y^L \gtrsim x \ .\]

Therefore, on hypergames
the idea is to define  both relations at the same time,  namely through  the greatest fixpoint of the following 
operator on pairs of relations:

\begin{defi}\label{ot}
Let $\Phi: {\mathcal P}({\mathcal H}\times {\mathcal H})\times {\mathcal P}({\mathcal H}\times {\mathcal H}) \longrightarrow {\mathcal P}({\mathcal H}\times {\mathcal H})\times {\mathcal P}({\mathcal H}\times {\mathcal H})$ be the operator defined by:
\begin{center}
$\hspace*{-1cm} \Phi ({\mathcal R}_1, {\mathcal R}_2)=(\{ (x,y)\  | \ \forall x^R.y {\mathcal R}_2 x^R \ \wedge\ \forall y^L. y^L {\mathcal R}_2 x \} , \ $
\\
\hspace*{6cm} $\{ (x,y)\  |\  \exists x^R.y {\mathcal R}_1 x^R \ \vee\ \exists y^L. y^L {\mathcal R}_1 x \} )$
\end{center}
\end{defi}

The above operator is monotone componentwise. Thus we can define: 

\begin{defi}\label{nov}
Let  the pair $(\mge, \nmge)$ be  the greatest fixpoint of $\Phi$.
\\
Furthermore, we define:\\
\begin{tabular}{lll}
 -- $x\mg y$\ \  & \mbox{iff}\ \   & $x\mge y \ \wedge\ y\nmge x$
\\ -- $x\bow y$\ \   & \mbox{iff}\ \  & $x\mge y \ \wedge\ y\mge x$
\\ -- $x || y$ \ \  & \mbox{iff}\ \  & $x\nmge y \ \wedge\ y\nmge x$
\end{tabular}
 \end{defi}
 
 As an immediate consequence of Tarski's Theorem, the above definition of the pair of relations $(\mge, \nmge)$ 
 as the greatest fixpoint of $\Phi$ gives us  \emph{Coinduction Principles}, which will be  useful in the sequel:

\begin{coin} We call \emph{$\Phi$-bisimulation} a pair of relations $ ({\mathcal R}_1, {\mathcal R}_2)$ such that
$({\mathcal R}_1, {\mathcal R}_2)\subseteq \Phi ({\mathcal R}_1, {\mathcal R}_2) $. The following principles hold:\medskip
 \begin{center}
  \reguno{({\mathcal R}_1, {\mathcal R}_2)\  \Phi\mbox{-bisimulation} \ \ \  x {\mathcal R}_1 y}{x \mge y} \ \ \ \ \ 
  \reguno{({\mathcal R}_1, {\mathcal R}_2)\  \Phi\mbox{-bisimulation}   \ \ \  x {\mathcal R}_2 y}{x \nmge y}
	 \end{center} 
	 \end{coin}
 
Notice that the pair of relations $(\gtrsim, \not\gtrsim)$ on Conway's games coincides with the restriction of 
$(\mge, \nmge)$ to Conway's games. Namely, the pair $(\gtrsim, \not\gtrsim)$
is  the least fixpoint of $\Phi$.

Moreover, somewhat surprisingly at  first sight, notice that the relations $\mge$ and $\nmge$ are \emph{not} disjoint. 
\emph{E.g.} the game $c=(\{ c\}, \{ c\})$
is such that both $c\mge 0$ and $c\nmge 0$ (and also $0\mge c$ and $0\nmge c$) hold. However, this is perfectly consistent in the
hypergame scenario, since
it  is in accordance with the fact that some
hypergames have non-losing strategies for more than one player. The following generalizes Conway's Characterization Theorem~\ref{car}:

\begin{thm}[Characterization] \label{chh}
Let $x$ be a hypergame. Then
\\
\begin{tabular}{llll}
$x\mg 0 $ \ \  \ \  & iff\ \ \ \ & $x$ has a non-losing strategy for L.
\\  $x\mn 0 $\ \  & iff & $x$ has a non-losing strategy for R.
\\ $x\bow 0 $ & iff & $x$ has a non-losing strategy for II.
\\ $x||0 $  &  iff & $x$ has a non-losing strategy for I.
\end{tabular}
\end{thm}
\begin{proof}
$(\Rightarrow)$ Assume $x\mg 0$, \emph{i.e.}  $x\mge 0$ and $0\nmge x$. 
We have to build non-losing strategies $f_{LI}$  for LI and $f_{LII}$ for LII. For LII:
since $x\mge 0$, then, by definition, $\forall x^R. 0 \nmge x^R$, \emph{i.e.}, for any R move $x^R$, $0\nmge x^R$.
 Let $x^R =(X^{RL}, X^{RR})$, then
$\exists x^{RL} \in X^{RL}. x^{RL} \mge 0$, that is there exists a L move $x^{RL}$ such that $x^{RL} \mge 0$. 
Hence, for any $x^R \in X^R$, we define $f_{LII} (x^R) = x^{RL}$.
Then, in order to extend  the definition of $f_{LII}$, we can apply again the two steps above starting from $X^{RL}$: either we  go on forever or we stop when R cannot move. In this way, we define $F_{LII}$ on coherent plays; on other plays $f_{LII}$ can be defined in an arbitrary way.
For LI: since $0\nmge x$, then $\exists x^L. x^L \mge 0$. This allows us to define $f_{LI}$ on the empty play. The definition of $f_{LI}$ can then be extended 
from $x^L$, using 
 the above construction of a non-losing strategy for LII.
\\ The other cases are dealt with similarly.
\\ $(\Leftarrow)$ We proceed  by coinduction, by showing all the four cases at the same time. Let 
\\ ${\mathcal R}_1= \{  (x,0) \ | \ x \mbox{ has a non-losing strategy for } LII \}\ \cup$
\\ \hspace*{5cm} $\{ (0,x) \ | \  x \mbox{ has a non-losing strategy for } RII  \} $, 
\\ ${\mathcal R}_2= \{  (x,0) \ | \ x \mbox{ has a non-losing strategy for } RI \}\ \cup$
\\ \hspace*{5cm} $\{ (0,x) \ | \  x \mbox{ has a non-losing strategy for } LI  \} $. 
\\ We  prove that $({\mathcal R}_1, {\mathcal R}_2)$ is a $\Phi$-bisimulation.  There are various cases to discuss. We only show one
case, the others being similar. We prove that, if $x {\mathcal R}_1 0$ and $x$ has a non-losing strategy for LII, then $\forall x^R. 0 {\mathcal R}_2
x^R$. If LII has a non-losing strategy on $x$, then,   for all $x^R$, there is a non-losing strategy for LI on $x^R$, hence, by definition,
$\forall x^R. 0 {\mathcal R}_2 x^R $.
\end{proof}

The following table summarizes the Characterization Theorem:
\begin{center}
\begin{tabular}{|c| l | l | l | l |}\hline
Non-losing strategies & \multicolumn{2}{c |}{Relations w.r.t. $0$} 
\\ \hline\hline
 L  &\  $x\mg 0$ \  &\ $ x\mge 0 \ \wedge\ 0\nmge x$\ 
\\ \hline
 R  &\ $x\mn 0$ \   &\ $ x\nmge 0 \ \wedge\ 0\mge x$ \
\\  \hline
\ \ II \ \  &\ $x\bow 0$ \   &\ $ x\mge 0 \ \wedge\ 0\mge x$ \
\\  \hline
I   &\ \ $x || 0$   &\ $ x\nmge 0 \ \wedge\ 0\nmge x$\  \\  \hline
\end{tabular}
\end{center}

\subsubsection{Properties of $\mge$.} \label{prom}
The following proposition, which  can be  proved by coinduction, 
generalizes  the corresponding results of \cite{Con76}  to hypergames:

\begin{prop}\label{bas}
For all hypergames $x,y$, we have
\\ \emph{\phantom i(i)} $x\mge y  \ \vee \ x\nmge y \ $.
\\ \emph{(ii)}  $x\nmge x^R\  \wedge \
x^L\nmge x\  \wedge \
x\mge x \  \wedge \ x\bow x\ . $
\end{prop}
\begin{proof}\hfill\\
\phantom i(i) If $x\mge y $ then we are done. Otherwise, assume $\neg x\mge y$. Then we show that $x\nmge y$. Namely, let
\[ \mathcal{R}_1 = \{ (x,y) \mid \neg x \nmge y \} \ \ \ \mbox{and}\  \ \  \mathcal{R}_2 = \{ (x,y) \mid \neg x \mge y \} \ .\]
One can easily check that $(\mathcal{R}_1, \mathcal{R}_2)$ is a $\Phi$-bisimulation. Therefore, in particular
$\neg x \mge y \ \Longrightarrow\ x\nmge y$.
\\ (ii) One can easily prove that the pair $(\mathcal{R}_1, \mathcal{R}_2)$ is a $\Phi$-bisimulation, where
\[ \mathcal{R}_1 = \{ (x,x) \mid  x \in \mathcal{H} \} \  \mbox{and}\  \mathcal{R}_2 = \{ (x,x^R) \mid x \in \mathcal{H}\} \cup
\{ (x^L,x) \mid x \in \mathcal{H}\} \ . \]
\end{proof}

As a consequence of Proposition~\ref{bas}(i), Theorem~\ref{chh} can be specialized as follows:

\begin{thm} \label{spe} Let $x$ be a hypergame. Then 
\\
\begin{tabular}{llll}
$x\mge 0 $ \ \  \ \  & iff\ \ \ \ & $x$ has a non-losing strategy for LII.
\\  $x\nmge 0 $\ \  & iff & $x$ has a non-losing strategy for RI.
\\ $0 \mge x $ & iff & $x$ has a non-losing strategy for RII.
\\ $0\nmge x $  &  iff & $x$ has a non-losing strategy for LI.
\end{tabular}
\end{thm}
\begin{proof}
Assume $x \mge 0 $. Then we show that $x$ has a non-losing strategy for LII. Namely, by Proposition~\ref{bas}(i), $0\mge  x$ or $0\nmge x$.
In the first case, $x\mge  0$ and $0\mge  x$, \emph{i.e.} $x\bow 0$ and,  by Theorem~\ref{chh}, $x$ has a non-losing strategy for II. In the latter case,
 $x\mge 0$ and $0\nmge x$, \emph{i.e.} $x \mg 0$ and, by Theorem~\ref{chh}, $x$ has a non-losing strategy for L. Vice versa, assume $x$ has a non-losing
 strategy for LII. Then we show that $x$ has a non-losing strategy for L or for II. Thus, by  Theorem~\ref{chh}, $x\mge 0$. Assume by contradiction that $x$ has no non-losing strategies for L and II. Then $x$ has no non-losing strategies for LI and RII, hence $x$ has no non-losing strategies for L,R,I,II, contradicting 
 Theorem~\ref{nls}. The other cases can be proven similarly.
\end{proof}

Contrary to what happens on  Conway's games, the relation $\mge$ is \emph{not} a partial order on
hypergames, since $\mge$ fails to be transitive; as a consequence,  $\bow$ is \emph{not} an equivalence.\medskip

\noindent \emph{Counterexample.} Let $c=(\{c\}, \{ c\})$ and $c_1=(\{c_1\} , \{ 0\})$.  Then $c\mge 0$, since $c$ has
 non-losing strategies for all the players. Moreover, one can show that $c_1 \mge c$, by coinduction, by
 considering the relations ${\mathcal R}_1 = \{ (c_1,c) \} \cup \{(0,c)\}$ and ${\mathcal R}_2 = \{ (c,c_1) \} \cup \{(c,0)\}$.
 Thus we have $c_1\mge c \ \wedge\ c\mge 0$. However, one can easily check that $c_1\mge 0$ does not hold.
\\
The problem is that  the  ``pivot'' $c$   allows for unlimited plays.
Transitivity can be recovered on pairs where the pivot is 
 a well-founded game. We omit the details.

\section{Operations on Games: Sum and Negation.}\label{gcon}
An important operation on games studied in \cite{Con76} is  sum, arising when more games are played simultaneously. 
There are various ways in which we can play several different games at once. We shall focus only on the most literal one, where
 at each step
 the next player selects any of the component games and makes any legal move on that game, the other
games remaining unchanged. The other player can either choose to move in the same component or in a different one.
 This kind of compound game  can be formalized through the \emph{(disjunctive) sum}, \cite{Con76}, which can be naturally extended to hypergames via the following coinductive definition, whereby the sum operation is obtained as final morphism:

\begin{defi}[Game Sum]\label{hs}
The \emph{sum on games} is given by the the final morphism $+: ({\mathcal H} \times {\mathcal H}, \alpha_{+}) \longrightarrow
({\mathcal H}, \mbox{id})$, where the coalgebra morphism $\alpha_{+}: {\mathcal H} \times {\mathcal H}\longrightarrow F({\mathcal H} \times {\mathcal H})$
is defined by 
\\
$ \alpha_+ (x,y)=  (\{ (x^L , y) \ | \ x^L \in X^L\} \cup \{(x,y^L) \ |\  y^L \in Y^L   \}, $\\
\hspace*{5.7cm}
$\{ (x^R , y) \ |\  x^R \in X^R\} \cup \{(x,y^R) \ | \ y^R \in Y^R   \} ) \ . $

\[
\xymatrix{
\mathcal{H}\times \mathcal{H} \ar[r]^{+}   \ar[d]_{\alpha_{+}} & \mathcal{H}   \ar[d]^{\mathit{id}}  \\
F( \mathcal{H}\times \mathcal{H}) \ar[r]_{F( +)}        & F(  \mathcal{H})          
}
\] 

\end{defi}

\noindent That is $+$ is such that:
\\
$ x+y = (\{ x^L + y \ | \ x^L \in X^L\} \cup \{  x+y^L \ |\  y^L \in Y^L   \},   $
\\ \hspace*{5.2cm} $\{ x^R + y \ |\  x^R \in X^R\} \cup \{  x+y^R \ | \ y^R \in Y^R   \} ) \ .$
\smallskip

The above definition of game sum subsumes  the definition of sum on Conway's games. Game sum 
resembles \emph{shuffling} on processes. In fact it coincides with interleaving in the case of impartial
games, \emph{i.e.} games on which L and R have the same options at any position. 

A typical example of a sum game arises when the two players play on two different chess boards at once, each time choosing a board on which to move, and performing a move on that board. Notice that in this way  the alternance of L and R in the single component games is missed. Clearly, this is not the way simultaneous chess exhibition
games are played. It is the way many end-games can be analyzed  in Go, for instance.  For other examples of sum games see Section~\ref{ih}, where generalized Nim and ``Traffic Jam'' games are
discussed.

Game sum satisfies usual properties of sum (commutativity, associativity, etc.) and it induces a commutative semigroup with the game $0$ as
zero. These properties  are established using the basic Coinduction Principle~\ref{cbas}. Equality  is hypergame identity.

\begin{prop}\label{coas}
For all games $x,y,z$,
\\ \emph{\phantom{ii}(i)} $x+0=x$
\\ \emph{\phantom i(ii)} $x+y= y+x$
\\ \emph{(iii)} $(x+y)+z= x+ (y+z).$
\end{prop}
\begin{proof}
By coinduction, showing that the symmetric closures of the relations $\mathcal{R}_1 = \{ (x+0,x) \mid x\in \mathcal {H} \}$, 
$\mathcal{R}_2 = \{ (x+y,y+x) \mid x,y\in \mathcal {H} \}$, $\mathcal{R}_3 = \{ ((x+y)+z, x+ (y+z)) \mid x,y,z\in \mathcal {H} \}$
are hyperbisimulations.
\end{proof}

The  \emph{negation} is a unary operation on games, which allows to build a new game, where the roles
of L and R are exchanged.The following is the coinductive extension to hypergames of   \emph{negation}  on Conway's games:

\begin{defi}[Game Negation]
The \emph{negation} of a game is given by the final morphism $-: ({\mathcal H} , \alpha_{-}) \longrightarrow
({\mathcal H}, \mbox{id})$, where the coalgebra morphism $\alpha_{-}: {\mathcal H} \longrightarrow F({\mathcal H})$
is defined by 
$ \alpha_{-} (x)=  (\{ x^R  \ | \ x^R \in X^R\}, \  \{   x^L \ |\  x^L \in X^L   \} ) \ . $
\[
\xymatrix{
\mathcal{H} \ar[r]^{-}   \ar[d]_{\alpha_{-}} & \mathcal{H}   \ar[d]^{\mathit{id}}  \\
F( \mathcal{H}) \ar[r]_{F( -)}        & F(  \mathcal{H})          
}
\] 
\end{defi}

\noindent That is:
$- x = (\{ -x^R\ | \ x^R \in X^R \} , \ \{-x^L\  |\  x^L \in X^L \} )  \ .$\smallskip

In particular, if $x$ has a non-losing/winning strategy for LI  (LII), then $-x$ has a non-losing/winning strategy for RI (RII), and symmetrically.
Taking seriously players L and R and not fixing a priori L or R to play first, makes the definition of $-$ very natural.

The following basic properties of negation  can be easily shown by building corresponding hyperbisimulations:

\begin{prop}
For all games $x,y$,
\\ \emph{\phantom i(i)} $ - (x+y) = -x + (-y)$
\\ \emph{(ii)} $-(-x)= x$.
\end{prop}

In what follows, we use the notation $x-y$ to denote $x+ (-y)$.

In Proposition~\ref{unoh} below, we summarize some properties of sum and negation w.r.t. the relation $\mge$, that subsume corresponding results on Conway's games. These properties will be useful in Section~\ref{ges}, where
we will discuss equivalences and congruences on games. 

We start by introducing \emph{equideterminacy},  a first natural equivalence on games induced by non-losing strategies. 

\begin{defi}[Equideterminacy]
Let $x,y$ be games. We say that $x$ and $y$ are \emph{equidetermined}, denoted by
$x\Updownarrow y$, whenever $x$ has a L (R,I,II) non-losing strategy if and only if $y$ has a L (R,I,II) non-losing strategy.
\end{defi}

Notice that equideterminacy divides the space of hypergames in the equivalence classes
of Fig.~\ref{fig}.

\begin{prop}
\label{unoh}
For all games $x,y,z$, 
\\ \emph{\phantom{ii}(i)}  $x-x \bow 0$.
\\ \emph{\phantom i(ii)}   $x\mge 0\ \wedge\ y\mge 0\ \Longrightarrow \ x+y\mge 0$.
\\ \emph{(iii)}  $x\mge y\ \Longleftrightarrow\ x-y\mge 0$.
\\ \emph{(iv)}   $x\mge y \ \Longrightarrow \  x+z \mge y+z$.
\\ \emph{\phantom i(v)}  $y\Updownarrow 0\ \Longrightarrow \ x+y\Updownarrow x$.
\\ \emph{(vi)}    $y-z\Updownarrow 0 \ \wedge z-z \Updownarrow 0  \ \Longrightarrow \  x+y\Updownarrow   x+z$.
\end{prop}
\begin{proof}

\noindent
\\ \phantom{ii}(i) By the Characterization Theorem~\ref{chh}, it is sufficient to prove that $x-x$ has a non-losing strategy for II. This is the  \emph{copy-cat strategy}, according to which, at each step,  player II simply
``copies'',   in the other component of   $x-x$, each move of player I in a component of $x-x$.
We omit the straightforward formalization of the copy-cat strategy.
\\ \phantom i(ii) By Theorem~\ref{spe}, it is sufficient to prove that 
LII has a non-losing strategy on $x+y$, assuming that LII has non-losing strategies on $x$ and $y$. A non-losing strategy for 
LII on $x + y$ can be obtained from the strategies on $x$ and $y$,
since L can always reply in the component where R moves in, by making a move according
to the non-losing strategy in that component.
\\ (iii) By Theorem~\ref{spe}, it is sufficient to prove that 
 $x\mge y$ iff  $x-y$ has a non-losing strategy for LII.  The implication $\Rightarrow$ follows by building directly a strategy for
LII, using the definition of $\mge$. For the converse implication, assume $x-y\mge 0$. Then, by Theorem~\ref{spe},  $x-y$ has a non-losing strategy for LII.
The thesis can then be proved by coinduction, showing that the  following pair of relations is a $\Phi$-bisimulation:
$\mathcal{R}_1 =\{ (x,y) \mid x-y \mbox{ has a non-losing strategy for LII } \}$, $\mathcal{R}_2 =\{ (x,y) \mid x-y \mbox{ has a non-losing strategy for LI } \}$.
\\ (iv) Assume $x\mge y$. Then, by item (iii) of this proposition, $x-y\mge 0$. By item (i), $z-z \mge 0$. Therefore, by item (ii), $x-y + (z-z) \mge 0$, \emph{i.e.}
$x+z - (y+z) \mge 0$, hence, by item (iii), $x+z \mge y+z$.
\\ \phantom i(v) From $y\Updownarrow 0$, by definition of $\Updownarrow$, it follows that $y$ has a winning strategy for II. We have to prove that  $x+y \Updownarrow x$.
Assume $x$ has a non-losing strategy for \emph{e.g.} L. Then one can easily
show that L has a non-losing strategy on $x+y$ as well, whereby L moves in $x$ according to the above non-losing strategy, and responds in $y$ 
to any move of R following the winning strategy for II which exists on $y$ by hypothesis. Vice versa assume by contradiction that L has a non-losing strategy on $x+y$ but no
non-losing strategy on $x$. Then, by Lemma~\ref{non}, RI or RII has a winning strategy on $x$. But then RI or RII has a winning strategy 
on $x+y$ as well, whereby R moves on $x$ according to the winning strategy and plays as II on $y$ according to the winning strategy which exists by hypothesis. 
By Lemma~\ref{non}, this contradicts the fact that L has a non-losing strategy on $x+y$. 
\\ (vi) Assume $y-z \Updownarrow 0$ and $z-z \Updownarrow 0$. Then, by item (v), since $y-z \Updownarrow 0$, $x+z + (y-z) \Updownarrow x+z$. Now $x+z + (y-z) = x+y + (z-z)$ and, since $z-z\Updownarrow 0$,
by item (v),
$x+y + (z-z) \Updownarrow x+y$. Hence, by transitivity of $\Updownarrow$, $x+y \Updownarrow x+z$.
\end{proof}


\begin{rem}\label{new}
\hfill
\\ \phantom{ii}(i) Notice that, on Conway's games, $x\Updownarrow 0$ is equivalent to $x\sim 0$, since, by Theorem~\ref{car},  $x\sim 0$ is equivalent to having a winning strategy
for player II on $x$.\\
However,  the implication $x \bow 0\ \Longrightarrow\ x\Updownarrow 0$ does \emph{not} hold on hypergames, since \emph{e.g.} $c \bow 0$, where $c = (\{ c\}, \{ c\} )$, because $c$ has a non-losing strategy for II, while $c\not\Updownarrow 0$, since $c$ has non-losing strategies for all players.
\\ \phantom i(ii) By  (i), 
 for Conway's games, items (v) and (vi) of Proposition~\ref{unoh} can be rephrased as 
\\ a. $y\sim 0 \ \Longrightarrow \ x+y \sim x$
\\ b. $y-z \sim 0 \ \wedge\ z-z\sim 0  \ \Longrightarrow \ x+y \sim x+z$.
\\ (iii) 
By Proposition~\ref{ceq} and  item (iv) of Proposition~\ref{unoh}, using commutativity of $+$, it follows that the equivalence $\sim$ on Conway's games is a \emph{congruence} w.r.t. sum.
\\ (iv)
Proposition~\ref{unoh} above gives an intuitive justification for why  $\sim$ is reflexive (and an equivalence relation on Conway's games), while $||$ is not. Namely,
by Proposition~\ref{unoh}, items (iii) and (i),  and by the Characterization Theorem~\ref{car}, we have that $x\sim x$ corresponds to having a winning strategy for player II on $x-x$. This is the very structural   ``copy-cat'' strategy, whereby   player II
copies any move of  player  I in the other component. The copy-cat strategy on $x-x$ is independent from the given game $x$. This is the reason why  player II has
a special role.  There is no counterpart of such a structural general strategy for  player I. The existence of a winning strategy for player I always depends on the specific nature of $x$.
\end{rem}

\section{Game Equivalences}\label{ges} 
Having defined games as elements of a final coalgebra, we have already taken games up-to bisimilarity equivalence, thus abstracting from 
superficial features of positions. Bisimilarity is a first structural equivalence on position graphs, but  on top of this one can define various notions of equivalences and congruences on games, by looking at strategies.  

A first notion of equivalence induced on games by strategies is equideterminacy $\Updownarrow$.
This is quite coarse, since  it  divides the space of hypergames in the equivalence classes
of Fig.~\ref{fig}.
 Trivially, $\Updownarrow$ is \emph{not} a congruence w.r.t. sum, already on Conway's games, since \emph{e.g.}, for the games $*1=(\{ 0\}, \{0\})$,  $*2=(\{ 0, *1\}, \{0,*1\})$,  we have:
$*1\Updownarrow *2$, but $*1+*1 \not\Updownarrow *2+*1$, because $*1+*1$ has  a winning strategy for  player II, while $*2+*1$ has a winning strategy for player I. 

The equivalence $\sim$ on Conway's games is a congruence (see Remark~\ref{new}(iii)), however, its extension $\bow$ to hypergames fails to be an equivalence. Therefore,
a question which naturally arises is whether there exists an equivalence on hypergames which extends $\sim$. 
The notion of \emph{contextual equivalence}   provides an answer to this question. This is defined as the closure under additive contexts of equideterminacy:

\begin{defi}[Contextual Equivalence]\hfill \label{conte}
\\
\phantom i(i) Let us consider the following class of \emph{additive contexts} on games:
\[ C[\ ] \ ::=\ [\ ]\  |\ C[\ ]  + x \ |\ x  +C[\ ]  \ , \]
where $x$ is a game.
\\ (ii) Let $\approx$ be the \emph{contextual equivalence} on games defined by:
\[ x\approx y \ \Longleftrightarrow \ \forall C[\ ].\ C[x]\Updownarrow C[y] \ . \]
\end{defi}

It is interesting to notice that, if $x\approx y$, then $x\Updownarrow y$. Moreover, if $\simeq$ is a congruence satisfying this property, then $\simeq \subseteq \approx$. \emph{I.e.}, we have:

\begin{lem}
The contextual equivalence $\approx$ is the greatest congruence refining equideterminacy.
\end{lem}
\begin{proof}
By definition, $\approx$ refines equideterminacy, \emph{i.e.} $x\approx y \ \Rightarrow\  x \Updownarrow y$, and $\approx$ is
a congruence, \emph{i.e.} $x\approx y \ \Rightarrow\  \forall C[\ ]. \ C[x] \approx C[y]$. Now assume that $\simeq$ is a congruence 
which refines equideterminacy. If $x\simeq y$, then $\forall C[\ ].\ C[x] \simeq C[y]$, since $\simeq$ is a congruence; moreover, since $\simeq$
refines equideterminacy, we have $\forall C[\ ].\ C[x] \Updownarrow C[y]$, hence $x\approx y$.
\end{proof}

As an immediate  consequence of commutativity and associativity of sum (Proposition~\ref{coas}), we have that the class of contexts in the definition of $\approx$ can be reduced as follows:

\begin{lem}\label{dd}
\[ x\approx y \ \Longleftrightarrow\ \forall D[\ ]. \ D[x]\Updownarrow D[y] \]
where $D[\ ]$ ranges over contexts of the shape $[\ ]+ z$, for $z$ any game.
\end{lem}

Finally, we show that the contextual equivalence $\approx$ coincides
with $\sim$ on Conway's games. 

\begin{thm}\label{aps} 
For all  Conway's games $x,y$,  
\[ x \approx y \ \Longleftrightarrow \  x\sim y \ .\]
\end{thm}
\begin{proof}

\noindent
$(\Rightarrow)$ Let $x\approx y$. Then  $x-y \Updownarrow y-y$. Moreover,  by  Proposition~\ref{unoh}(i), $y-y \sim 0$, 
hence  $y-y \Updownarrow 0$, by Remark~\ref{new}(i). Therefore, by transitivity of $\Updownarrow$, 
$x-y \Updownarrow 0$. Hence, by Remark~\ref{new}(i),
 $x-y\sim 0$. Finally, by Proposition~\ref{unoh}(iii), $x\sim y$. 

\noindent
$(\Leftarrow)$ Let $x\sim y$. Then, by Proposition~\ref{unoh}(iii), $x- y\sim 0$ and, by Proposition~\ref{unoh}(i), $y-y \sim 0$.
Hence, by Remark~\ref{new}(i), $x-y \Updownarrow 0$ and $y-y \Updownarrow 0$. Then, by Proposition~\ref{unoh}(vi),
for any hypergame $z$, $z+x\Updownarrow z+y$, thus $x\approx y$.
\end{proof}

Clearly, the extension of the above theorem to hypergames fails, since $\bow$ is not an equivalence. However, 
as we will see in the next section, when restricted to impartial hypergames, the contextual equivalence turns out to be the equivalence
induced by the generalized Grundy semantics.

\section{The Theory of Impartial Games}
\label{img}
In this section, we focus on the subclass  of \emph{impartial games}, where, at each position, L and R have the same moves.
Such games  can be simply represented by $x=X$, where $X$ is the set of moves (for L or R). Coalgebraically, this amounts to saying  that impartial games are the elements of  the final coalgebra $\mathcal{J}$ of the powerset functor.
Thus they correspond directly  to possibly non-wellfounded sets, and they can be represented as finite or infinite, possibly cyclic graphs,
having a node for each position of
the game, and a direct edge from $x$ to $y$ when it is legal to move from $x$ to $y$. The subclass $\mathcal{I}$ of impartial Conway's games correspond
to well-founded sets and they form an initial  algebra of the powerset functor.

In this section, first we specialize some general results to the case of impartial hypergames, including Determinacy and Characterization Theorems. Then, in Section~\ref{sui}, we recall
the Grundy-Sprague theory for 
dealing with impartial Conway's games, based on a  class of \emph{canonical
games}. We show that these give, via the Grundy-Sprague function, 
 a \emph{compositional semantics} of games, that  induces exactly
the contextual equivalence $\approx$. In Section~\ref{ih}, we extend the above  theory to hypergames, by revisiting in a coalgebraic setting Smith's generalization of Grundy-Sprague results. In particular, we introduce a class of \emph{canonical
hypergames}, extending the \emph{Nim numbers}. Moreover,
we show how to extend the \emph{semantics} given by the Grundy function  for impartial Conway's  games to the whole class of impartial
hypergames. Our approach extends other approaches in the literature, where the generalized Grundy function is defined only on certain classes of finite 
cyclic graphs, see \emph{e.g.} \cite{FR01}. 
We illustrate our results on an  example.\bigskip

 Since on impartial games the distinction between L and R is blurred, we can only consider player I and player II, and Theorem~\ref{qua} specializes as follows:
 
 \begin{thm}
 Any impartial game has a  winning strategy  either for player I or for   player II or the two players can draw.
 \end{thm}
 \begin{proof}
 The proof follows from Theorem~\ref{qua}, by observing that, on impartial games, L has a non-losing strategy iff R has, and neither L nor R can have a winning strategy,
 because otherwise we would have both L and R having a winning strategy, contradicting  Theorem~\ref{qua}. 
 \end{proof}

Moreover, 
the following lemma holds:

\begin{lem}\label{lcd}
Let $x,y\in \mathcal{J}$. Then
\[ (x\mge y \ \Longleftrightarrow\ y\mge x \ \Longleftrightarrow\  x\bow y) \ \wedge\  (x\nmge y \ \Longleftrightarrow\ y\nmge x \ \Longleftrightarrow\  x|| y) \]
\end{lem}
\begin{proof}
The proof is by coinduction, by showing that the following pair of relations is a $\Phi$-bisimulation:
$\mathcal{R}_1 = \{ (y,x) \mid x\mge y \} $ and  $\mathcal{R}_2 = \{ (y,x) \mid x\nmge y\} $. In order to prove that $(\mathcal{R}_1, \mathcal{R}_2)$ is a $\Phi$-bisimulation,
we assume  $(y,x) \in \mathcal{R}_1$ and we prove that $\forall  x^L . \ (x^L,y) \in \mathcal{R}_2$ and $\forall  y^R . \ ( x,y^R) \in \mathcal{R}_2$.  We omit the proof for  $\mathcal{R}_2$ which is similar. Assume then $(y,x) \in \mathcal{R}_1$. Hence $x\mge y$, and by definition of $\mge$, 
$\forall  x^R . \ y \nmge  x^R$ and $\forall  y^L . \ y^L\nmge x$. Hence, since $X^L=X^R$ and $Y^L=Y^R$, we have 
$\forall  x^L . \ y \nmge  x^L$ and $\forall  y^R . \ y^R\nmge x$. Thus, by definition of $\mathcal{R}_2$, 
$\forall  x^L . \ ( x^L,y) \in \mathcal{R}_2$ and $\forall  y^R . \ (x, y^R) \in \mathcal{R}_2$.
\end{proof}

Hence, by Theorem~\ref{spe} and Lemma~\ref{lc},  we have:

\begin{thm}[Characterization of Impartial Hypergames]\label{chs}
Let $x\in \mathcal{J}$. Then 
\\
\begin{tabular}{lll}
$x\mge 0 $ \ \ \ \   & iff\ \ \ \ & $x$ has a non-losing strategy for II.
\\  $x\nmge 0 $ \ \ \ \   & iff\ \ \ \ & $x$ has a non-losing strategy for I.
\\ $x\mge 0 $ and not $ x \nmge 0$ \ \ \ \   & iff\ \ \ \ & $x$ has a winning strategy for II.
\\ $x\nmge 0 $ and not $ x \mge 0$ \ \ \ \   & iff\ \ \ \ & $x$ has a winning strategy for I.
\end{tabular}
\end{thm}

\subsection{Impartial Conway's Games.} \label{sui}
Impartial well-founded games are dealt with the theory of Grundy-Sprague, \cite{Gru39,Spra35}.
 Central to this theory
 is \emph{Nim}, a well-known impartial game, which is played with a
number of heaps of matchsticks. The legal move is to strictly decrease the number of matchsticks in any heap (and throw away the removed sticks). A player unable to move because no sticks remain is the loser. 

The Nim game with one heap of size $n$ can be represented as the Conway game $*n$, defined (inductively) by 
\[ *n= \{ *0, *1, \ldots , *(n-1)\} \ ,\]
where $*0=0$.

Namely,  with a heap of size $n$,  the options of the next player consist in moving to a heap of size $0, 1, \ldots , n-1$.
The number $n$ is called the \emph{Grundy number} of the game.
Clearly, if $n=0$,  player II wins, otherwise player I has a winning strategy, moving to $*0$. 

Nim games can be naturally extended to heaps of arbitrary ordinal length, and they
 correspond precisely  to von Neumann ordinals, \emph{i.e.} $*\alpha = \{ *\beta \mid \beta < \alpha \}$.

Notice that different Nim games are told apart by $\sim$:

\begin{lem}\label{ta}
Let $*\alpha, *\beta$ be Nim games. Then
\[ *\alpha \sim *\beta \ \Longleftrightarrow\  \alpha = \beta \ .\]
\end{lem}
\begin{proof}
If $\alpha\neq \beta$, then $*\alpha \not\approx *\beta$, since $*\alpha+ *\alpha$ has a winning strategy for II (the copy-cat strategy), while one can check that
$*\alpha+ *\beta$ has a winning strategy for I. Hence, by Theorem~\ref{aps}, $*\alpha \not\sim *\beta$.
\end{proof}

Nim games are central in game theory, since there is a classical result (by Grundy and Sprague, independently, 
\cite{Gru39,Spra35}) showing that any
impartial well-founded  game ``behaves'' as a Nim game, in the sense that it is $\sim$-equivalent  to a single-heap Nim game (see \cite{Con76}, Chapter 11). 
The algorithm for discovering the Nim game (or the Grundy number) corresponding to a given impartial game $x$ proceeds
inductively as follows. \medskip

\noindent \emph{\emph{\bf Grundy Algorithm.}
Assume that the Grundy numbers of the positions in  $x$ are $\{ \alpha_0, \alpha_1 ,\ldots \}$, then the Grundy number of $x$ is
the \emph{minimal excludent} (\emph{mex}) of $\{ \alpha_0, \alpha_1 ,\ldots \}$, where
the mex of a set $X\subsetneq \mathit{Ord}$ is the least ordinal in the complement  set  of $X$.}\medskip

Formally, the Grundy-Sprague function $g: \mathcal{I}\rightarrow \mathit{Ord}$, associating to each impartial Conway's game the corresponding Grundy number via the mex algorithm, amounts to the following algebra morphism:

\begin{prop} \label{tau}
Let $ \mathit{mex}:\mathcal{P}( \mathit{Ord}) \rightarrow \mathit{Ord}$ be the mex function. Then the \emph{Grundy-Sprague function} $g$ is the unique morphism from the
initial $\mathcal{P}$-algebra $(\mathcal{I}, \mathit{id})$ to the $\mathcal{P}$-algebra $(\mathit{Ord}, \mathit{mex})$:
\[
\xymatrix{
\mathcal{I} \ar[r]^{g}  & \mathit{Ord}  \\
\mathcal{P}( \mathcal{I}) \ar[r]_{\mathcal{P}( g)}      \ar[u]^{\mathit{id}}     & \mathcal{P}( \mathit{Ord}) \ar[u]_{\mathit{mex}}           
}
\] 
\end{prop}

The following holds:

\begin{thm}[\cite{Gru39,Spra35,Con76}]\hfill\label{sgs}
\\ \emph{\phantom ii)} Nim games are canonical, in the sense that,
for any game $x\in \mathcal{I}$, the Grundy function $g$ gives the unique  Nim game $*g(x)$  such that $x\sim *g(x)$.
\\ \emph{ii)} Any impartial game $x$ has a winning strategy for  player II  if and only if $g(x)$ is $0$, otherwise it has a winning strategy for  player I.
\end{thm}

\subsubsection*{Sum of impartial games.} 
In~\cite{Gru39,Spra35}, an efficient algorithm for computing the Grundy number corresponding to the sum of impartial games is provided, based on
binary sum without carries. Namely, for all numbers $n_1, n_2$, one can define the \emph{Nim sum} of $n_1$ and $n_2$ by $n_1 \oplus n_2 =n$, where 
$n$ is the number resulting from the binary sum without carries of $n_1 $ and $n_2$. \emph{E.g.} $1\oplus 3 =2$, since the binary sum without carries of 10 and 11 is 10.

The Nim sum satisfies the following property:

\begin{prop}[\cite{Gru39,Spra35}]\label{cinse}
\label{seq} For all Nim games $*\alpha, *\beta$,
\[ *\alpha + *\beta \sim *(\alpha \oplus \beta) \ .\]
\end{prop}

As a consequence of the above proposition and of Lemma~\ref{ta}, we have: 

\begin{cor}
\label{compo} For all $x,y \in \mathcal{I}$,
\[ g(x+y) = g(x) \oplus g(y) \ .\]
\end{cor}
\begin{proof}
By Theorem~\ref{sgs}(i), $*g(x+y) \sim x+y$, then by Theorem~\ref{sgs}(i) and congruence of $\sim$, $x+y \sim  *g(x) + * g(y)$, hence by Proposition~\ref{cinse}
and transitivity of $\sim$, $*g(x+y) \sim * (g(x) \oplus g(y))$. Finally, by Lemma~\ref{ta}, $*g(x+y) = * (g(x) \oplus g(y))$.
\end{proof}

By the above corollary, the Grundy  number corresponding to $x+y$ is obtained by Nim summing the Grundy numbers of $x$ and $y$. 
The Nim sum is efficient when dealing with finite Grundy numbers.
It applies also to infinite ordinals, but clearly for ordinals above limit ordinals, \emph{i.e.} $\lambda$ 
such that $2^{\lambda} =\lambda$, it does not provide substantial improvements w.r.t. the direct mex calculation of the sum game.

\subsubsection*{Fully abstract semantics of impartial Conway's games.}   By Corollary~\ref{compo}, the Grundy function $g$  provides a compositional semantics
on  impartial  Conway's games. This semantics 
is  \emph{fully abstract}  w.r.t. contextual equivalence, namely we have: 

 \begin{thm}[Full Abstraction] For all $x,y \in \mathcal{I}$, 
\[ g(x)=g(y)  \  \Longleftrightarrow\  x\approx y   \]
\end{thm}

\begin{proof}

\noindent
$(\Rightarrow)$ Let $g(x)=g(y)$. By Theorem~\ref{sgs}(i), using transitivity of $\sim$, we have $x\sim y$. Hence, by Theorem~\ref{aps}, $x\approx y$.

\noindent
$(\Leftarrow)$ Let $x\approx y$. By Theorem~\ref{aps} and Theorem~\ref{sgs}(i), using  transitivity of $\approx$, we have $*g(x)\approx *g(y)$. But then
$g(x)=g(y)$, by Lemma~\ref{ta}.
\end{proof}

\subsection{Impartial Hypergames} \label{ih}

We introduce a class of \emph{canonical hypergames}, extending the Nim games:

\begin{defi}[Canonical Hypergames]\label{canh}
The \emph{canonical hypergames} $ *\sigma$  are defined by:
\begin{iteMize}{$\bullet$}
\item the Nim games $*\alpha$, for $\alpha \in \mathit{Ord}$;
 \item the hypergames $*\infty = \{ *\infty \}$ and $*\infty_K= \{ *\infty\} \cup \{ *k \mid  k\in K\} $, for $K\neq \emptyset $, $K\subseteq \mathit{Ord}$.
\end{iteMize}
\noindent The \emph{generalized Grundy numbers} are the $\sigma$'s such that $*\sigma$  is a canonical hypergame.
\end{defi}
In what follows, we denote  by $\sigma, \tau$ generalized Grundy numbers, and by $\alpha, \beta$ ordinal (well-founded) Grundy numbers.
We will also use the notation $*\infty_{\emptyset}$ ($\infty_{\emptyset}$) to denote the hypergame
(Grundy number) $*\infty$  ($\infty$).

The following lemma holds on canonical hypergames:

\begin{lem}\label{lc}
The game $*\alpha$ is a  win for player II if and only if $\alpha$ is $0$, otherwise it is a win for player I. An hypergame
 $*\infty_K$ is a  win for  player I if and only if $0\in K$, otherwise both players have non-losing strategies.
\end{lem}
\begin{proof}
We only have to prove the thesis for the hypergames
$*\infty_K$, since for well-founded games $*\alpha$ the result follows from Theorem~\ref{sgs}. 
Using Theorem~\ref{chs}, the thesis on hypergames follows
by showing  that:
\\ (a) any hypergame $*\infty_{K}$ is such that   $*\infty_{K} \nmge 0$;
\\ (b) the hypergame $*\infty_{K}$ has no subscript $0$ iff  $*\infty_{K} \mge 0$.
\\ 
$(a)$ 
First of all, notice that, since $*\infty$ has a non-losing strategy for II, then, by Theorem~\ref{chs} and Lemma~\ref{lcd}, 
$*\infty$ is such that $0\mge *\infty$.
  Then, by Definition of $\nmge$,  we have
$*\infty_{K}\nmge 0$.
 \\ 
 $(b\Rightarrow)$ Assume $0\not\in K$. Then 
 $*\infty_K \mge 0$, since,   for all elements  $k\in K$,  $0\nmge k$, by Lemma~\ref{ta} and Proposition~\ref{bas}(i),
 and $0\nmge *\infty$, because $*\infty$ has a non-losing
 strategy for I.
 \\
 $(b\Rightarrow)$ Assume $*\infty_K\mge 0$. If by contradiction $0\in K$, then $*\infty_K \mge 0$ does not
 hold, since $0\nmge 0$ does not hold.
\end{proof}

In the following lemma, we prove that different canonical hypergames are told apart by the contextual equivalence; this extends Lemma~\ref{ta} to  hypergames:

\begin{lem}\label{ugu}
For all canonical hypergames  $*\sigma, *\tau$, 
\[ *\sigma \approx *\tau \ \Longleftrightarrow\  \sigma = \tau \ .\]
\end{lem}
\begin{proof}
Let $\sigma \neq \tau$. Then we show that there is a context $D[\ ]$ such that  $D[*\sigma] \not\Updownarrow
 D[*\tau]$.
There are various cases to consider, according to the shape of $*\sigma, *\tau$. If both $\sigma, \tau \in \mathit{Ord}$, then
the thesis follows from Lemma~\ref{ta} and Theorem~\ref{aps}.
  If $\sigma=\infty_H$, $\tau=\infty_K$, $H\neq K$,  then there exists 
$\alpha\in H\setminus K$ (or $\alpha\in K\setminus H$).  If $\alpha\neq 0$, then 
 we 
get that $*\infty_H + *\alpha $ has a winning strategy for player I, which opens  on $*\infty_H$ moving to $\alpha$, and then plays as II according to the copy-cat strategy.
While one can easily check that $*\infty_K + *\alpha $ has non-losing strategies for both players. 
If $\alpha =0$, then $*\sigma \not\Updownarrow *\tau$, since, by Lemma~\ref{lc}, 
 player II has a non-losing strategy on H (K) but not on K (H).
If $\sigma\in \mathit{Ord}$ and $\tau =\infty_K$, then if $\sigma=0$ or $0\not\in K$, then, by Lemma~\ref{lc},
$*\sigma\not\Updownarrow *\tau$. Finally, if $\sigma=\alpha$, $\alpha \neq 0$, and
$\tau=\infty_K$, $0\in K$, then for $D[\ ] = [\ ] + *\alpha$, we have that $D[*\sigma]$ has a winning strategy for II, hence no non-losing strategies for I, while $D[*\tau]$ has a non-losing strategy for I.
\end{proof}

 The Nim sum  can be extended to the whole class of generalized Grundy numbers: 
 \begin{defi}[Generalized Nim Sum] \label{gns}
The \emph{generalized Nim sum} $\oplus$ is defined by extending the Nim sum 
 on  $\infty$-Grundy numbers by:
 \[ \alpha\oplus \infty_K = \infty_K \oplus \alpha = \infty_{\{ k\oplus \alpha \ |\ k\in K \}} \ \ \ \ \ \infty_K \oplus \infty_H =\infty \ .\]
 \end{defi}
 
 The  following extends Proposition~\ref{seq} to the case of hypergames:
 
 \begin{prop}\label{csum}
For all canonical hypergames $*\sigma, *\tau$, 
\[ *\sigma + *\tau \approx *(\sigma \oplus \tau) \ .\]
\end{prop}
\begin{proof}
If $\sigma, \tau\in \mathit{Ord}$, then the thesis follows by Proposition~\ref{seq} and Theorem~\ref{aps}. Then assume that $\sigma =\infty_K$ and $\tau =\alpha$,
or $\sigma =\infty_K$ and $\tau =\infty_H$. We have to show that $\forall z\in \mathcal{J}.\ (*\sigma + *\tau) + z \Updownarrow *(\sigma\oplus \tau) +z$.
By Theorem~\ref{chs},  it is sufficient to prove that 
\[
(*\sigma + *\tau) +z \mge 0 \ \Longleftrightarrow\ * (\sigma \oplus \tau) +z \mge 0 \ \  \mbox{ and } \ \ 
(*\sigma + *\tau) +z \nmge 0 \ \Longleftrightarrow\ * (\sigma \oplus \tau) +z \nmge 0 \ . \ \ (a)
\]
\noindent The two implications $(\Rightarrow)$ in (a) follow, using Lemma~\ref{lcd},  by proving that the following pair of relations form a $\Phi$-bisimulation:
\\ $\mathcal{R}_1 = \{ (*(\sigma\oplus \tau)+z,0 ) \mid z\in \mathcal{J} \ \wedge\ (*\sigma + *\tau)+z \mge 0\}$
\\ $\mathcal{R}_2 = \{ (0, *(\sigma\oplus \tau)+z) \mid z\in \mathcal{J} \ \wedge\  0 \nmge (*\sigma + *\tau)+z \}$.
\\ Namely, let $(*(\sigma\oplus \tau)+z,0 )\in \mathcal{R}_1 $. We need to prove that $\forall x\in *(\sigma \oplus \tau)+z.\ (0,x) \in \mathcal{R}_2$.
Two cases arise: (i) $x= *(\sigma\oplus\tau) +z'$ and $z'\in z$, or (ii) $x=x'+z$ and $x'\in *(\sigma \oplus \tau)$. 
\\ In case (i), since $(*\sigma + *\tau) +z \mge 0$,
then 0\nmge $(*\sigma +*\tau) +z' $, hence $(0, *(\sigma + \tau) + z')\in \mathcal{R}_2$.
\\ In case (ii), if $\sigma= \infty_K$ and $\tau=\alpha$, then $x'\in * (\sigma \oplus \tau)$ iff (a) $x'= *\infty$ or (b) $x'= *(k\oplus \alpha)$, for $k\in K$. 
If $x'= *\infty$, we have to prove that $(0, *\infty +z)\in \mathcal{R}_2$. But, since $\infty= \infty \oplus\infty$, it is sufficient to prove that $0\nmge *(\infty \oplus\infty ) +z $.
But this can be  easily proved.  If $x'= *(k\oplus \alpha)$, for $k\in K$, then we have to prove that $0\nmge (*k + * \alpha) +z $.
This follows from $(*\infty_K + * \alpha) +z \mge 0$, by definition of $\mge$, since $*k\in * \infty_K$. If $\sigma=\infty_K$ and $\tau =\infty_H$, the proof is similar.
\\ Using similar arguments one can also prove that if $(0,*(\sigma\oplus \tau)+z)\in \mathcal{R}_2$, then $\exists x\in *(\sigma\oplus \tau) +z.\ (0,x) \in \mathcal{R}_1$.
\\ Vice versa, in order to prove the implications $(\Leftarrow)$ in (a), it is sufficient to show that the following  pair of relations form a $\Phi$-bisimulation:
\\ $\mathcal{R}_1 = \{ ((*\sigma + *\tau)+z,0 ) \mid z\in \mathcal{J} \ \wedge\   *(\sigma\oplus \tau)+z  \mge 0\}$
\\ $\mathcal{R}_2 = \{ (0, (*\sigma + *\tau)+z ) \mid z\in \mathcal{J} \ \wedge\   0 \nmge *(\sigma\oplus \tau)+z  \}$.
\\ We omit the proof which uses arguments similar to the ones above.
 \end{proof}

\subsubsection*{Generalized Grundy function.}
We define a \emph{generalized Grundy function} $\gamma$, associating to each hypergame $x$ a generalized  Grundy
number such that $x\approx *\gamma(x)$. As we  will show, as a consequence,  such function provides a \emph{compositional semantics} on hypergames, inducing
the contextual equivalence $\approx$.

We define the generalized Grundy function $\gamma: \mathcal{J}\rightarrow \mathit{Ord}\cup \{ \infty_K  \mid K \subseteq \mathit{Ord} \}$ in two steps.
\begin{enumerate}[(1)]
\item  First, we define $\gamma_0:\mathcal{J}\rightarrow \mathit{Ord}\cup \{ \bot \}$ as least fixpoint of a suitable monotone operator. 
The function $\gamma_0$
will mark all positions $\approx$-equivalent to a  Nim game with the corresponding Grundy number, and it marks as $\bot$ all positions corresponding to
$\infty$-hypergames.
More precisely,  $\gamma_0: \mathcal{J}\rightarrow \mathit{Ord}\cup \{ \bot \}$ is defined as limit of  subsequent approximations, starting from the overall undefined function $f_0= \lambda x. \bot$. Approximations of $\gamma_0$ are
built by induction: at a given step, the next approximation function will mark a position $x$ with the $\mathit{mex}$ $\alpha$ of the successors of $x$ which already have received a marking
in $\mathit{Ord}$ in previous steps, provided that any $\bot$-successor $y\in x$ has a successor $z\in y$ which is already marked by  $\alpha$.
Intuitively, this marking procedure of a position $x$ with a $\mathit{mex}$ $\alpha$ will remain ``correct'' in the future, since any successor $y\in x$, which has not yet received a marking in $\mathit{Ord}$, will never receive $\alpha$ as marking, because it has a successor already marked by $\alpha$; thus $\alpha$ will remain the $\mathit{mex}$ of
the successors of $x$ marked in $\mathit{Ord}$ at any subsequent step.
\item Once we have defined $\gamma_0$, we can  define the generalized Grundy function $\gamma$,  by suitably marking $\bot$-positions with 
$\infty$-Grundy numbers.
\end{enumerate}

In what follows, we formalize the construction hinted above.

\begin{defi}
Let $(\mathit{Ord}\cup \{ \bot \}, \sqsubseteq)$ be the c.p.o. endowed with the flat ordering  such that $\bot \sqsubseteq \alpha$, for any $\alpha\in \mathit{Ord}$.
\\ \phantom{ii}i) Let $([\mathcal{J}\rightarrow \mathit{Ord}\cup \{ \bot \}], \sqsubseteq)$ denote the space of functions  $f:\mathcal{J}\rightarrow \mathit{Ord}\cup \{ \bot \}$
endowed  with the point-wise ordering on functions.
\\ \phantom iii) Let $f\in [\mathcal{J}\rightarrow \mathit{Ord}\cup \{ \bot \}]$ and  let $x\in \mathcal{J}$, we define $\mathcal{F}_f (x) =
\{ fy \mid y\in x \ \wedge fy\neq \bot \}$.
\\ iii)  Let  $D \subseteq
[\mathcal{J}\rightarrow \mathit{Ord}\cup \{ \bot \}]$ denote the subspace of functions $f$ such that, for all $x\in \mathcal{J}$, 
\[ fx=\alpha\in \mathit{Ord} 
\  \Longrightarrow \  \alpha = \mathit{mex}  (\mathcal{F}_f (x))  \  \wedge \   \forall y\in x\ (fy=\bot  \ \Rightarrow \
 \exists z\in y. \ fz=
\mathit{mex} (\mathcal{F}_f (x))) \ . \]
\end{defi}

\begin{lem}\label{dicio}
 $(D, \sqsubseteq)$ with the point-wise ordering on functions is a c.p.o.
\end{lem}

\begin{defi}
Let $T:(D, \sqsubseteq) \longrightarrow  (D, \sqsubseteq)$ be the function operator defined by
\[ T(f)(x)= \begin{cases}
\mathit{mex}  (\mathcal{F}_f (x)) & \mbox{ if }\  \forall y\in x\ (fy=\bot  \ \Longrightarrow \
 \exists z\in y. \ fz=
\mathit{mex} (\mathcal{F}_f (x)))
\\
\bot & \mbox{ otherwise }
\end{cases} \]
\end{defi}

\begin{lem}\label{ven}
The operator $T$ is monotone over $D$.
\end{lem}
\begin{proof}
Let $f\sqsubseteq f'$. Assume by contradiction that $T(f)\not \sqsubseteq T(f')$. Then we necessarily have that $T(f)(x)=  \mathit{mex} (\mathcal{F}_f (x))=\alpha$ and
$T(f')(x)=  \mathit{mex}(\mathcal{F}_{f'}  (x))=\alpha'$ with $\alpha\neq \alpha'$. Since $f\sqsubseteq f'$, this implies $\alpha< \alpha'$. Then there is $y\in x$ such that $f'y=\alpha$,
while $fy=\bot$.  But then
there  is $z\in y$  such that $fz= \alpha$. Thus, since   $f\sqsubseteq f'$, also $f'z=\alpha$. This contradicts the fact that  $f'y=\alpha$  is a $\mathit{mex}$.
\end{proof}

By Lemmata~\ref{dicio} and~\ref{ven}, we have:

\begin{thm}\label{lfix} The operator $T$ has least fixpoint
 $\gamma_0= \bigsqcup_{n\in \omega} f_{\alpha}$, where 
 \[ f_0 = \la x. \bot \ \ \ \ \  \  \ \ f_{\alpha+1} = Tf_{\alpha}\ \ \  \ \ \ \ \ f_{\lambda} = \bigsqcup_{\alpha<\lambda} Tf_{\alpha} \ , \mbox{ for } \lambda \mbox{ limit.}\]
\end{thm}

Theorem~\ref{lfix} above is an instance of a general fixpoint theorem. However, notice that, in this case, we have to pay attention to the fact that we are dealing with
proper classes. There are various way outs. The most obvious rests on restricting on sets whose transitive closure has cardinality smaller than an inaccessible cardinal.
 Other approaches rely on the fact that the iterations are point-wise eventually constant. Namely, the
 value of the fixpoint on a given set is determined after a number of iterations corresponding to the cardinality  of its
 transitive closure, which is always a set.

By construction, $\gamma_0$ satisfies the following property:

\begin{prop}
For any $x\in {\mathcal J}$,
\[ \gamma_0 (x) =\alpha \in \mathit{Ord} \ \Longrightarrow\ \alpha= \mathit{mex} \{ \gamma_0(y) \mid y\in x \ \wedge\ \gamma(y) \in \mathit{Ord} \} \ .\]
\end{prop}

Now we are ready to define the generalized Grundy function $\gamma$, which is obtained from $\gamma_0$ by extending the marking on $\perp$-positions. These 
will receive $\infty$-marking as follows:

\begin{defi}[Generalized Grundy Function]\label{geng}
Let $\gamma: \mathcal{J}\rightarrow \mathit{Ord}\cup \{ \infty_K \mid K\in \mathit{Ord} \}$ be defined by 
\[ \gamma(x) = \begin{cases}
\gamma_0(x) & \mbox{ if } \gamma_0(x) \in \mathit{Ord}
\\
\infty_K & \mbox{ for } K =\{ \gamma(y)  \mid y\in x \ \wedge\  \gamma_0(y)\neq \bot \}, \mbox{  otherwise }
\end{cases}\]
\end{defi}

\begin{rem} Notice that any fixpoint of the operator $T$ induces, via Definition~\ref{geng},  a function which satisfies the axiomatization of generalized Grundy function given in the literature, see 
\emph{e.g.} \cite{FR01}. It has been shown that, for some classes of cyclic games, \emph{i.e.} the \emph{locally path-bounded} graphs,
a generalized Grundy function exists uniquely with values in $\omega \cup \{ \infty_K \}_{K\in \omega}$, \cite{FR01}. A graph is locally path-bounded if for every node 
$x$ there is  $n_x\in\omega$ such that the length  of every path consisting of different nodes, starting from $x$, does not exceed $n_x$. 
The above result of~\cite{FR01} can be strengthened, by observing that bisimilar graphs have the same Grundy function. Therefore, a Grundy function 
with values in $\mathit{Ord} \cup \{ \infty_K \}_{K\in \omega}$ exists uniquely on all graphs which
are bisimilar to locally path-bounded graphs. However, notice that this is not the largest class of graphs for which a Grundy function with values in 
$\mathit{Ord} \cup \{ \infty_K \}_{K\in \omega}$ exists uniquely. Namely, there are path-finite but not locally-path bounded graphs for which such property holds, \emph{e.g.}
the following:
\vspace*{1ex}
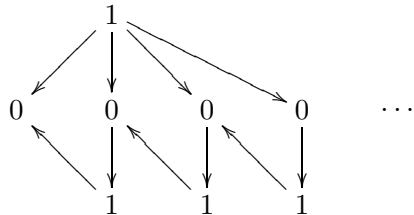
\[
 \xymatrixrowsep={0.001pc}
\xymatrixcolsep={0.0000001pc}
 \xymatrix@!=1.5pc{
 & & 1 \ar[ddll]  \ar[dd]  \ar[ddrr] \ar[ddrrrr]&& &&  & &
\\
&&&&&&&&&&
\\
 0  & & 0 \ar[dd]  & & 0 \ar[dd] & & 0 \ar[dd]  &&  \ldots
\\
&&&&&&&&&&
\\
 & &1 \ar[uull] & & 1 \ar[lluu] & &1  \ar[lluu] &  & 
 } 
\] 

In this section, we have proved that a generalized Grundy function with values in $\mathit{Ord} \cup \{ \infty_K \}_{K\in \mathit{Ord}}$ exists
on the whole class of hypergames. However, this is \emph{not} unique, because  $T$
 \emph{does not} have a unique fixpoint. 
As a consequence, on hypergames, there exists more functions satisfying the axiomatization of the Grundy function in~\cite{FR01}.
Namely, let us consider the following infinite hypergame, on which the function $\gamma$ induced by the least fixpoint of $T$,  $\gamma_0$,  provides the marking:\vspace*{1ex}
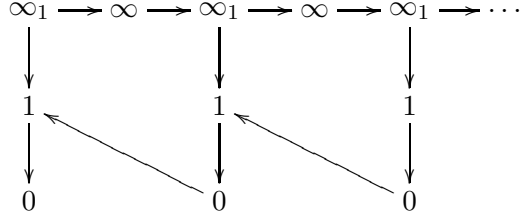
\[
 \xymatrixrowsep={0.001pc}
\xymatrixcolsep={0.0000001pc}
 \xymatrix@!=1.5pc{
\infty_1 \ar[rr]  \ar[dd] & & \infty   \ar[rr] & & \infty_1  \ar[rr]  \ar[dd] & & \infty   \ar[rr] & &  \infty_1 \ar[rr]  \ar[dd] & & \ldots
\\
&&&&&&&&&&
\\
 1  \ar[dd] & & & & 1 \ar[dd]  & &  && 1\ar[dd] && 
\\
&&&&&&&&&&
\\
0 & & & & 0 \ar[lllluu] & & &  & 0 \ar[lllluu] & &
 } 
\] 
One can check that there are other  fixpoints of $T$, such as \emph{e.g.}  the function  $\gamma'_0$,
which on the hypergame above  gives the following marking:\vspace*{1ex}
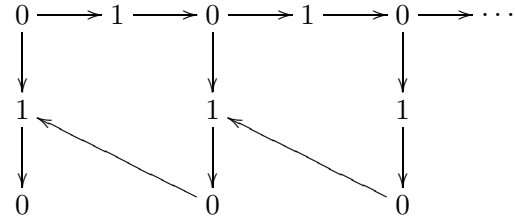
\[
 \xymatrixrowsep={0.001pc}
\xymatrixcolsep={0.0000001pc}
 \xymatrix@!=1.5pc{
0 \ar[rr]  \ar[dd] & & 1   \ar[rr] & & 0  \ar[rr]  \ar[dd] & & 1 \ar[rr] & &  0 \ar[rr]  \ar[dd] & & \ldots
\\
&&&&&&&&&&
\\
 1  \ar[dd] & & & & 1 \ar[dd]  & &  && 1\ar[dd] && 
\\
&&&&&&&&&&
\\
0 & & & & 0 \ar[lllluu] & & &  & 0 \ar[lllluu] & &
 } 
\] 
Thus also $\gamma'_0$ would satisfy the axiomatic definition of the Grundy function.
However, if we denote by $\gamma'$ the function induced by $\gamma'_0$ via Definition~\ref{geng}, this provides a marking
which is \emph{not} safe, since it is not the case that $x\approx *\gamma' (x)$ for any $x$. Namely, 
 some nodes marked with $0$ in the above figure have a non-losing strategy for  player I. On the contrary, the least fixpoint $\gamma_0$ provides a safe marking, since, as we will prove,  $x\approx *\gamma(x)$ for any $x$. Finally, it is interesting to notice that, if we add the natural requirement $\forall x.\ x\approx *\gamma(x)$ in the axiomatization
 of the Grundy function, then, by Lemma~\ref{ugu}, the Grundy function exists  uniquely  on the whole class of hypergames.
\end{rem}

\subsubsection*{Properties of $\gamma$.}\label{propg}
Here we show  that $\gamma$ is ``well-behaved'', \emph{i.e.}, for any hypergame $x$, $x\approx *\gamma(x)$, and it provides a \emph{compositional semantics} of hypergames, \emph{fully abstract} w.r.t. the contextual 
 equivalence $\approx$.

\begin{thm}\label{jam}
For any $x\in \mathcal{J}$, we have
\[ x\approx *\gamma(x)\ .\]
\end{thm}
\begin{proof} 
By Theorems~\ref{aps} and~\ref{chs}, it is sufficient to prove that,
 for any  $z\in \mathcal{J}$, 
\[ x+z \mge 0 \ \Longleftrightarrow *\gamma(x) +z \mge 0 \mbox{    and    } x+z \nmge 0 \ \Longleftrightarrow *\gamma(x) +z \nmge 0 \ .\]
\noindent In order to prove  $(\Leftarrow)$ implications, it is sufficient to show that the following pair of relations is a $\Phi$-bisimulation:
\\ $\mathcal{R}_1 = \{ (x+z, 0) \mid  z\in \mathcal{J} \ \wedge\ *\gamma (x) +z \mge 0 \} \ \ \ \  \mathcal{R}_2 = \{ (0, x+z) \mid z\in \mathcal{J} \ \wedge\ 0\nmge *\gamma (x) +z \}$.
\\  Vice versa, in order to prove  $(\Rightarrow)$ implications, it is sufficient to show that the following pair of relations is a $\Phi$-bisimulation:
\\ $\mathcal{R}'_1 = \{ (*\gamma (x)+z, 0) \mid z\in \mathcal{J} \ \wedge\ x +z \mge 0 \} \ \ \ \  \mathcal{R}'_2 = \{ (0,*\gamma (x)+z) \mid z\in \mathcal{J} \ \wedge\ 0 \nmge x+z \}$.
\\ Both proofs proceed by case analysis. We only show that if $(x+z,0)\in \mathcal{R}_1$, then $\forall y\in x+z.\ (0,y) \in \mathcal{R}_2$. 
The other cases being dealt with similarly.
If $y\in x+z$ is such that $y= x+z'$ and $z'\in z$, the thesis
follows from the fact that $0\nmge * \gamma(x) +z'$, since 
$*\gamma(x) +z \mge 0$. If $y=x'+z$ and $x'\in x$, then if $\gamma(x') \in \mathit{Ord}$, by definition of
$\gamma$, $*\gamma (x') \in *\gamma(x)$, and hence $0\nmge *\gamma(x') +z $; thus $(0, x'+z) \in \mathcal{R}_2$. If $\gamma(x') = \infty_K$, then one can show that $0\nmge
*\infty_K +z$ using Theorem~\ref{chs}, since player I has always a non-losing strategy on $*\infty_K +z$, for any $z$. Hence $(0,x'+z)\in \mathcal{R}_2$.
\end{proof}

By Theorem~\ref{jam} and Proposition~\ref{csum}, we have:

\begin{prop}[Compositionality]
For all $x,y \in \mathcal{J}$,
\[  \gamma (x+y) = \gamma(x) \oplus \gamma(y) \ .\]
\end{prop}

\begin{thm}[Full Abstraction]\label{fad}
For all $x,y \in \mathcal{J}$,
\[ \gamma(x)=\gamma(y) \ \Longleftrightarrow\ x\approx y \ .\]
\end{thm}
\begin{proof}
 $(\Rightarrow)$ If  $\gamma(x)=\gamma(y)$, then $*\gamma(x) \approx *\gamma(y)$. Thus, using Theorem~\ref{jam}, we have
$x \approx *\gamma(x) \approx *\gamma(y) \approx y$.
\\ $(\Leftarrow)$  If $x\approx y$, by Theorem~\ref{jam}, we have  $*\gamma(x)\approx x\approx y \approx *\gamma(y)$, thus  
$*\gamma(x)\approx *\gamma(y)$. Hence, by Lemma~\ref{ugu},  $*\gamma(x)= *\gamma(y)$. 
\end{proof}

\subsubsection*{A famous motivating example: Traffic Jams.} 
For finite hypergames, the computation of the   function $\gamma$ of Definition~\ref{geng} is effective.
      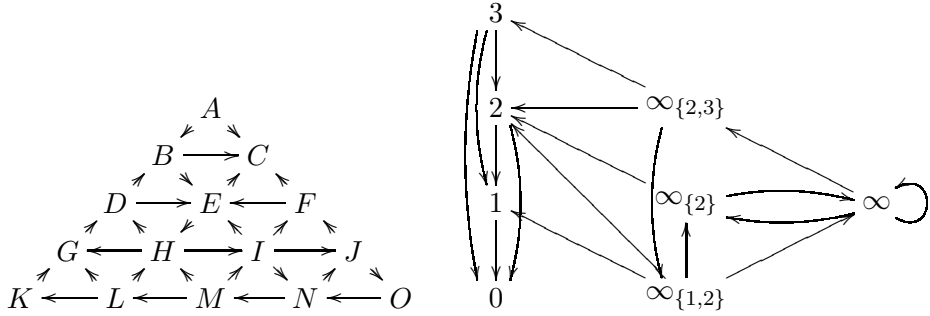
\begin{figure}
   \[
  \xymatrixrowsep={0.001pc} 
\xymatrixcolsep={0.0000001pc}
 \xymatrix@!=1.5pc{
 &&&&&&&&&& 
 3  \ar@/_/@<-1.2ex>[dddddd]  \ar@/_/@<-0.3ex>[dddd]   \ar[dd]  && &&  &&&&
 \\
 &&&&&&&&&&
  && && &&&&
 \\
 & & &  & A  \ar[ld]  \ar[rd]& & & & &
 &
   2   \ar@/^/@<0.7ex>[dddd]   \ar[dd] && && \infty_{\{2,3\}}  \ar[llll]  \ar[lllluu]  \ar@/_/@<-1.5ex>[dddd]   &&&&
  \\ 
 & & & B  \ar[rr] \ar[rd]& & C & & & &
 &
  && &&  &&&&
 \\
 & & D \ar[ru] \ar[rr] & & E \ar[ru] \ar[ld]& & F \ar[ll] \ar[lu]& & &
 &
   1 \ar[dd] && &&  \ar[uullll] \infty_{\{2\}}  \ar@/^/@<-0.3ex>[rrrr]      &&&&  \infty\ar@(dr,ur)[] \ar[uullll]  \ar@/^/@<0.3ex>[llll]  
 \\
 & G \ar[ru]& & H \ar[ll] \ar[lu] \ar[rr]& & I \ar[lu]\ar[ru]\ar[rr]\ar[rd]&& J \ar[lu]\ar[rd]& &
 &
  && &&  &&&&
 \\
 K \ar[ru]& & L \ar[ll]\ar[lu]\ar[ru]&& M\ar[ll]\ar[lu]\ar[ru] && N \ar[ll]\ar[ru] && O\ar[ll] & \hspace*{1.5cm}
 &
  0  && && \infty_{\{1,2\}} \ar[uurrrr]  \ar[uu]  \ar[lllluu]  \ar[lllluuuu]   &&&&
 } 
 \]
   \caption{A traffic jam game graph and the corresponding marked hypergame.}
   \label{gra} 
    \end{figure}
    Following \cite{Con76}, 
 consider the  concrete game,  corresponding to  the game graph  in the left-hand part of Fig.~\ref{gra}. 
Think of the graph as the map of a fictitious country, where nodes correspond to  towns, and edges represent
motorways between them. The initial position of the game corresponds to the town where  a vehicle is initially placed.
Each player has to move such vehicle to a next town along the motorway. If this is not possible, then the player loses.
The hypergame corresponding to the game graph in the left-hand part of Fig.~\ref{gra} appears in the righ-thand part, together with 
the marking given by $\gamma$. Positions C,D,K collapse to 0, positions A,E,G collapse to 1, positions B,F,H collapse to 2, position L corresponds to 3,
positions N,O collapse to $\infty$, position I corresponds to $\infty_{\{ 1,2\}}$, position J corresponds to $\infty_{\{ 2\}}$, and position M corresponds to
$\infty_{\{ 2,3\}}$.
 By Lemma~\ref{lc}, we can immediately check  which player has a non-losing strategy in any original position.
 
 Now, let us consider a version of the above game with more than one vehicle, and where at each step the current player chooses a vehicle to move, 
 assuming that each town is big enough to accommodate all vehicles at once.  This game corresponds to the sum of the games
 with single vehicles. In order to compute non-losing strategies for the sum game, one  can use the {generalized Nim sum},
 as defined in Definition~\ref{gns}.
 If, for example,  we have vehicles at positions H and I in Fig.~\ref{gra}, then the game is winning for player I, since
H corresponds to  2 in the hypergame and I to  $\infty_{\{1,2\}}$, and 
$2\oplus  \infty_{\{1,2\}} =\infty_{\{ 2\oplus 1, 2\oplus 2\}} = \infty_{\{3,0\}}$. While a game with vehicles in I and J is a draw, since 
J corresponds to $\infty_{\{ 2\}}$ and 
$\infty_{\{1,2\}} \oplus
 \infty_{\{2\}} =\infty$. The game where infinitely many vehicles are parked in non-0 positions  is a draw.


\subsubsection*{Efficient characterizations of the contextual equivalence.}
Using the above results, one can check  that, in the case of impartial hypergames, we can simplify the class of contexts in the definition of the contextual equivalence $\approx$,
by considering only \emph{well-founded canonical hypergames}:

\begin{cor} \label{corv} Let $x,y \in  \mathcal{J}$, then
\[ x\approx y \ \Longleftrightarrow \  \forall \! *\!\alpha \mbox{ well-founded canonical game. } x+*\alpha \Updownarrow y+ *\alpha \ .\]
\end{cor}
\begin{proof}
$(\Rightarrow)$ Immediate.
\\  $(\Leftarrow)$ By Theorem~\ref{fad},  it is sufficient to show that, if $\gamma(x) \neq \gamma(y)$, then there exists $\alpha\in \mathit{Ord}$ such that $x+ *\alpha\not\Updownarrow  y+ *\alpha$.
To this aim, by Theorem~\ref{jam} and the fact that  $\approx$ is a congruence, it is sufficient to show that $*\gamma(x) + *\alpha \not \Updownarrow *\gamma(y) + *\alpha$,
for some $\alpha\in \mathit{Ord}$. This can be easily shown by case analysis on $\gamma(x), \gamma(y)$, using the hypothesis $\gamma(x)\neq \gamma(y)$.
\end{proof}

The following definition allows us to identify ``well-behaved'' hypergames, and to formulate an alternative efficient
 characterization of the contextual equivalence on impartial hypergames:

\begin{defi}\label{tren}
Let $x$ be a hypergame. We define $x\Downarrow $ iff $x-x \Updownarrow 0$.
\end{defi}

Notice  that the above definition is given  for \emph{all} hypergames. In  the impartial case, $x-x$ coincides with $x+x$, and being well-behaved amounts to having Grundy number in $\mathit{Ord}$. 
 Conway's games are clearly well-behaved. The following characterization of $\approx$ on impartial hypergames holds:

\begin{prop}\label{corcon}
Let $x,y \in \mathcal{J}$, then 
\[ x\approx y \ \Longleftrightarrow\ (x\Downarrow \ \wedge\ y\Downarrow \ \wedge \ x+y \Updownarrow 0) \ \vee \ 
(x\not \Downarrow \ \wedge\ y\not\Downarrow \ \wedge\ \forall *\alpha. \ x+*\alpha \Updownarrow y+*\alpha ) \ .\]
\end{prop}
\begin{proof}
Since $x\Downarrow$ iff $\gamma(x)\in\mathit{Ord}$, it is easy to check that, if  $x\approx y$, then either $x\Downarrow \ \wedge\ y\Downarrow $ or $x\not \Downarrow \ \wedge\ y\not\Downarrow $. Let $\gamma(x), \gamma(y)\in\mathit{Ord}$.
Then $x\approx y$ iff $*\gamma(x)\approx *\gamma(y)$, by Theorem~\ref{jam}, iff
$*\gamma(x)\sim *\gamma(y)$, by Theorem~\ref{aps},  iff $*\gamma(x)+ *\gamma(y)\sim 0$, by Proposition~\ref{unoh}(iii), iff
$*\gamma(x)+ *\gamma(y)\Updownarrow 0$, since $*\gamma(x), *\gamma(y)$ are Conway's games,
iff $x+y\Updownarrow 0$, by Theorem~\ref{jam} and congruence of $\approx$.
If both $x\not \Downarrow $ and $ y\not\Downarrow $, then the thesis follows from Corollary~\ref{corv}.
\end{proof}

According to the above characterization,
$\approx$ discriminates between well-behaved and non-wellbehaved hypergames. 
Checking whether two well-behaved hypergames $x,y$ are equivalent is particularly efficient, since a single game, $x+y$, has to be considered. For 
non-wellbehaved  
 hypergames, it is sufficient to test $x,y$ in well-founded canonical contexts.
Efficient characterizations of the contextual equivalence for the whole class of hypergames, based on Definition~\ref{tren}, are explored in~\cite{HLR11}.

\section{Comparison with Related Work and Directions for Future Work}\label{final}
\subsubsection*{Loopy games.} 
The theory  of general mixed loopy games, where infinite plays can be either winning for L or for R, or draws, appears very complex.
For instance, already for the case of fixed impartial games (where no draws are admitted), determinacy fails if the Axiom of Choice is 
assumed (see \emph{e.g.} \cite{Jech03}, Lemma 33.1). Notice that some encoding is necessary to cast in Conway's setting  such a result.
Namely, the straightforward representation of  the infinite game on $\omega^{\omega}$  would collapse to $\infty$, \emph{i.e.}
 the self-singleton.

In \cite{BCG82}, Chapter 11, 
a $\gtrsim$ relation is introduced for fixed loopy games, which  is proved to be transitive. It  allows to  approximate  the behavior of a loopy game, possibly with
finite games. However, this technique works only if certain fixpoints exist. Such theory has been later further developed and revisited in  other works, see \emph{e.g.} \cite{San02,San02a}.

Our theory allows to deal with the class of games where infinite plays are draws, namely free games in \cite{BCG82} terminology, in a quite general and comprehensive way.
In~\cite{HLR11a}, we use algebraic and coalgebraic methods to account  for more general classes of games, such as mixed loopy games.

\subsubsection*{Games and automata.}
The notion of hypergame that we have investigated in this paper is related to 
 the notion of infinite game considered in the automata theoretic approach, originating in work
of Church, B\"uchi, McNaughton and Rabin (see \emph{e.g.} \cite{Tho02}).
In this approach, games are defined by the  graphs of positions.  L and R have  different positions, in general, but 
 L is always taken as first player. Only games with infinite  plays are usually considered there. 
 These games are fixed, \emph{i.e.} no draws are admitted.
Winning strategies are connected with automata, and also the problem of a (efficient) computation of such strategies
is considered. Recently,
non-losing strategies have been considered also in this setting,   \emph{e.g.}  in the context of model checking
for the $\mu$-calculus, see \cite{GLLS07}. We plan to pursue this approach. Notice that also in this case some encoding
is necessary for representing such games as Conway's games to avoid the collapse into the self-singleton $\infty$.

\subsubsection*{Games for semantics of logics and programming languages.} 
Game Semantics was introduced in the early 90's in the construction of the first fully complete model of Classical Multiplicative Linear Logic \cite{AJ94},
 and  of the first syntax-independent fully abstract 
model of PCF, by Abramsky-Jagadeesan-Malacaria, Hyland-Ong, and Nickau, independently.
  Game Semantics has been used for modeling 
a variety of programming languages and logical systems, and more recently
 for  applications in computer-assisted verification and program analysis, \cite{AGMO03}.
In Game Semantics, 2-player games are considered, which can be encoded as Conway games, despite the different presentation. 
For more details see \cite{AJ94}.  The key difference between the Game Semantics approach and our approach lies in 
the crucial definition of Conway sum, which does not necessarily imply that projections in component games are correct plays,
\emph{i.e.} strictly alternating. On the contrary, all the operations on games, \emph{e.g.} tensor product, linear implication, satisfy
this condition. 
Furthermore, in Game Semantics, infinite plays are
always considered as winning for one of the two players, as in the  case of  fixed games. A general framework based on coalgebras,  encompassing both
Conway's games and games used in Game Semantics, as well as other games, is presented in~\cite{HLR11a}. Such unifying framework allows us to
capture the common nature of games arising in different settings.

\subsubsection*{Traced categories of games.} In \cite{Joy77}, Joyal showed how Conway (finite) games and winning strategies can be endowed with
 a structure of a traced category. This  provides an alternate account of the contextual equivalence $\approx$; namely, existence of a morphism between the corresponding games. When hypergames and non-losing strategies are considered, 
  Joyal's
 categorical construction apparently does not work, since we lose closure under composition (this is related to the fact that our relation $\mge$ is 
 not transitive). In~\cite{HLR11}, generalizations of Joyal's category to hypergames are investigated, based on the new notion \emph{balanced} non-losing 
 strategy. Interestingly, when this construction is extended to mixed loopy games (see~\cite{HLR11a}), the categorical equivalence coincides with the loopy
 equivalence of~\cite{BCG82}. However, a categorical construction capturing the contextual equivalence on hypergames is still missing.
 
\subsubsection*{Games and coalgebras.} In  \cite{BM96}, a simple coalgebraic notion of game is introduced and utilized. It is folklore that 
 bisimilarity can be  defined  as a 2-player game, where one player tries to prove
bisimilarity, while  the other tries to disprove it,  see \emph{e.g.} 
\cite{BM96}. After some encoding, this game turns out to be a fixed game in the sense of \cite{BCG82}, where infinite plays are
winning for the 
player who tries to prove bisimilarity.

Notice that all our notions and methods could be given on graphs rather than graphs up-to bisimilarity. The former approach  can be convenient in applications, because computing the minimal graph with respect to bisimilarity could be complex, however the latter 
approach is conceptually more perspicuous. 

\subsubsection*{Conumbers.} Conway's numbers \cite{Con76} amount  to   Conway's games $x$ such that no member of $X^L$ is 
$\mge$ any member of $X^R$, and all positions of a number are numbers. Thus, once we have defined hypergames and the relations 
$\mge$, $\nmge$,
 we can define the subclass of
\emph{conumbers},
 together with suitable operations extending those on numbers. It would be interesting to investigate the properties
 of such a class of hypergames. An intriguing point is whether it is possible to define a partial order, since, as seen in this paper,
  the relation $\mge$  is 
 not transitive on hypergames.

\subsubsection*{Compound games.} 
In this paper, we have considered the (disjunctive) sum for building compound games. However, there are several different  ways of combining games, which are analyzed in \cite{Con76}, Chapter 14, for the case of finite games. It would be interesting to investigate such theory of compound games in the setting of hypergames. Similarly
for the mis\`ere situation, where the winner is the player who \emph{does not} perform the last move.

\subsubsection*{Game equivalences.} 
Equivalences on games and hypergames have been extensively investigated in this paper. The contextual
equivalence has been introduced and studied.  In particular,
Proposition~\ref{corcon} gives an alternative efficient characterization of the contextual equivalence for impartial hypergames.
For general partizan hypergames, where we do not have an analogue of Grundy semantics, alternative efficient characterizations of contextual equivalence would
be even more useful. This issue is investigated  in~\cite{HLR11}.

  \subsubsection*{Towards canonical forms for partizan games.}
 The theory of impartial games, as shown in Section~\ref{img} is quite nice. In particular, impartial games admit canonical forms given by generalized Grundy numbers. A question
which naturally arises is about canonical forms for general partizan games. 
In \cite{Con76}, Chapter 10, Conway studies  \emph{canonical forms} of general partizan games, and provides a  technique for
reducing a game to its canonical form, which works for finite  games, \emph{i.e.} games with only finitely many positions. This technique consists in simplifying 
a game, by 
eliminating all \emph{dominated} and \emph{reversible} positions. This  provides canonical forms of Conway's games, which
one can show  to coincide with Grundy numbers, in the case of impartial games. The extension of the above technique to infinite Conway's  games and more generally
to hypergames appears to be  problematic. 
A naive application of the same procedure to impartial hypergames fails, even in the case
of  hypergames representable by a finite graph and with finite Grundy number. We leave it as
an open problem to investigate a  generalization of the simplification procedure for   hypergames.


\begin{thebibliography}{DDDD88}    

\bibitem[AGMO03]{AGMO03} S. Abramsky, D.R. Ghica, A.S. Murawski and C.-H.L. Ong. Applying Game Semantics to Compositional Software Modeling and Verifications,  Proc. of  \emph{TACAS 2004}, Springer LNCS {\bf 2988}, 2004, 421--435.
                                      
                                      
\bibitem[AJ94]{AJ94} S. Abramsky, R. Jagadesaan. Games and Full Completeness for Multiplicative Linear logic,  Journal of Symbolic Logic
 {\bf 59}, 1994, 543--574.                                      
                                      
\bibitem[Acz88]{Acz88}
 P. Aczel.  \emph{Non-wellfounded sets}, CSLI Lecture Notes {\bf 14},
 Stanford 1988.
 
 \bibitem[BM96]{BM96} J. Barwise, L. Moss. Vicious Circles, CSLI Lecture Notes {\bf 60}, Stanford 1996.
 
 \bibitem[BCG82]{BCG82} E. Berlekamp, J. Conway, R. Guy. Winning Ways, Academic Press, 1982.

\bibitem[Con76]{Con76} J.H. Conway. On Numbers and Games, second edition, A K Peters Ltd,   2001 (first edition by Academic Press, 1976). 

\bibitem[FH83]{FH83} M. Forti, F. Honsell. Set-theory with 
free construction principles, \emph{Ann. Scuola Norm. Sup. Pisa}, 
Cl. Sci. (4){\bf 10}, 1983, 493--522.

\bibitem[FR01]{FR01} A. Fraenkel, O. Rahat. Infinite cyclic impartial games, \emph{Theoretical Computer Science},  {\bf 252}, 2001, 13--22. 


\bibitem[GLLS07]{GLLS07} O. Grumberg, M. Lange, M. Leucker, S. Shoham. When Not Losing Is Better than Winning: Abstraction and Refinement
for the Full $\mu$-calculus, Information and Computation {\bf 205(8)}, 2007, 1130--1148.

\bibitem[Gru39]{Gru39} P.M. Grundy. Mathematics and games, \emph{Eureka}, 
{\bf 2}, 1939, 6--8.


\bibitem[HL09]{HL09} F. Honsell, M. Lenisa. Conway Games, coalgebraically, Proc. of \emph{CALCO'09},  Springer LNCS {\bf 5728}, 2009, 300--316.

\bibitem[HLR11]{HLR11} F. Honsell, M. Lenisa, R. Redamalla. Equivalences and Congruences on Infinite Conway's Games, submitted for publication, 2011.

\bibitem[HLR11a]{HLR11a} F. Honsell, M. Lenisa, R. Redamalla. A General Framework for Games, submitted for publication, 2011.

\bibitem[Jech03]{Jech03} T. Jech. Set Theory, third edition, Springer, 2003.

\bibitem[Joy77]{Joy77}  A. Joyal. Remarques sur la Theorie des Jeux
a deux personnes, Gazette des sciences mathematiques du 
Quebec  {\bf 1(4)}, 1977 (English translation by R. Houston, 2003).
  

\bibitem[San02]{San02} L. Santocanale, Free $\mu$-lattices, \emph{J. Pure Appl. Algebra} {\bf 168}, 2002, 227--264.

\bibitem[San02a]{San02a} L. Santocanale,  $\mu$-bicomplete categories and parity games, \emph{Theor. Inform. Appl.} {\bf 36(2)}, 2002, 195--227.

\bibitem[Smi66]{Smi66} C.A.B. Smith. Graphs and composite games,   \emph{J. Combin. Th.}
{\bf 1}, 1966, 51--81.

\bibitem[Spra35]{Spra35} R.P. Sprague. \"Uber mathematische Kampfspiele,  \emph{Tohoku Math. J.}
{\bf 41}, 1935-6, 438--444.


\bibitem[Tho02]{Tho02} W. Thomas. Infinite games and verification, Proc. of  \emph{CAV'02}, Springer LNCS {\bf 2404}, 2002, 58--64. 



\end{thebibliography}
\end{document}